\documentclass[journal]{IEEEtran}

\usepackage{graphicx}
\usepackage{subfig}
\usepackage{multirow}
\usepackage{array}
\newcolumntype{P}[1]{>{\centering\arraybackslash}p{#1}}
\newcolumntype{M}[1]{>{\centering\arraybackslash}m{#1}}

\usepackage{amsthm,amssymb,amsmath,multirow,color,amsfonts}%
\usepackage[update,prepend]{epstopdf}
\usepackage{multirow}
\usepackage[latin1]{inputenc}
\usepackage{tikz}
\usepackage{bbm} 
\usepackage{pdfpages}
\usepackage{multirow}
\usepackage{comment}
\newtheorem{proposition}{Proposition}

\usepackage{cite}

\usepackage{textcomp}
\usepackage{psfrag}
\usepackage{arydshln}
\usepackage{url}
\usepackage{soul}
\usepackage[nolist]{acronym}
\usepackage{algorithm,algorithmic}

\usepackage{mathtools,lipsum}
\usepackage{cuted}
\usepackage[capitalise]{cleveref}
\Crefname{equation}{Eq.\!}{Eqs.\!}
\Crefname{figure}{Fig.\!}{Figs.\!}
\Crefname{tabular}{Tab.\!}{Tabs.\!}
\Crefname{section}{Section\!}{Sections.\!}

\usepackage{xr}
\makeatletter
\newcommand*{\addFileDependency}[1]{
  \typeout{(#1)}
  \@addtofilelist{#1}
  \IfFileExists{#1}{}{\typeout{No file #1.}}
}
\makeatother

\def\nb0{{\mathbf{0}}}
\def\nb1{{\mathbf{1}}}

\newtheorem{lemma}{Lemma}

\newtheorem{theorem}{Theorem}

\newtheorem{corollary}{Corollary}

\newtheorem{remark}{Remark}
\newtheorem{assumption}{Assumption}
	
\begin{document}
\graphicspath{{./Figures/}}
	\begin{acronym}

\acro{5G-NR}{5G New Radio}
\acro{3GPP}{3rd Generation Partnership Project}
\acro{ABS}{aerial base station}
\acro{AC}{address coding}
\acro{ACF}{autocorrelation function}
\acro{ACR}{autocorrelation receiver}
\acro{ADC}{analog-to-digital converter}
\acrodef{aic}[AIC]{Analog-to-Information Converter}     
\acro{AIC}[AIC]{Akaike information criterion}
\acro{aric}[ARIC]{asymmetric restricted isometry constant}
\acro{arip}[ARIP]{asymmetric restricted isometry property}

\acro{ARQ}{Automatic Repeat Request}
\acro{AUB}{asymptotic union bound}
\acrodef{awgn}[AWGN]{Additive White Gaussian Noise}     
\acro{AWGN}{additive white Gaussian noise}

\acro{APSK}[PSK]{asymmetric PSK} 

\acro{waric}[AWRICs]{asymmetric weak restricted isometry constants}
\acro{warip}[AWRIP]{asymmetric weak restricted isometry property}
\acro{BCH}{Bose, Chaudhuri, and Hocquenghem}        
\acro{BCHC}[BCHSC]{BCH based source coding}
\acro{BEP}{bit error probability}
\acro{BFC}{block fading channel}
\acro{BG}[BG]{Bernoulli-Gaussian}
\acro{BGG}{Bernoulli-Generalized Gaussian}
\acro{BPAM}{binary pulse amplitude modulation}
\acro{BPDN}{Basis Pursuit Denoising}
\acro{BPPM}{binary pulse position modulation}
\acro{BPSK}{Binary Phase Shift Keying}
\acro{BPZF}{bandpass zonal filter}
\acro{BSC}{binary symmetric channels}              
\acro{BU}[BU]{Bernoulli-uniform}
\acro{BER}{bit error rate}
\acro{BS}{base station}
\acro{BW}{BandWidth}
\acro{BLLL}{ binary log-linear learning }

\acro{CP}{Cyclic Prefix}
\acrodef{cdf}[CDF]{cumulative distribution function}   
\acro{CDF}{Cumulative Distribution Function}
\acrodef{c.d.f.}[CDF]{cumulative distribution function}
\acro{CCDF}{complementary cumulative distribution function}
\acrodef{ccdf}[CCDF]{complementary CDF}               
\acrodef{c.c.d.f.}[CCDF]{complementary cumulative distribution function}
\acro{CD}{cooperative diversity}

\acro{CDMA}{Code Division Multiple Access}
\acro{ch.f.}{characteristic function}
\acro{CIR}{channel impulse response}
\acro{cosamp}[CoSaMP]{compressive sampling matching pursuit}
\acro{CR}{cognitive radio}
\acro{cs}[CS]{compressed sensing}                   
\acrodef{cscapital}[CS]{Compressed sensing} 
\acrodef{CS}[CS]{compressed sensing}
\acro{CSI}{channel state information}
\acro{CCSDS}{consultative committee for space data systems}
\acro{CC}{convolutional coding}
\acro{Covid19}[COVID-19]{Coronavirus disease}

\acro{DAA}{detect and avoid}
\acro{DAB}{digital audio broadcasting}
\acro{DCT}{discrete cosine transform}
\acro{dft}[DFT]{discrete Fourier transform}
\acro{DR}{distortion-rate}
\acro{DS}{direct sequence}
\acro{DS-SS}{direct-sequence spread-spectrum}
\acro{DTR}{differential transmitted-reference}
\acro{DVB-H}{digital video broadcasting\,--\,handheld}
\acro{DVB-T}{digital video broadcasting\,--\,terrestrial}
\acro{DL}{DownLink}
\acro{DSSS}{Direct Sequence Spread Spectrum}
\acro{DFT-s-OFDM}{Discrete Fourier Transform-spread-Orthogonal Frequency Division Multiplexing}
\acro{DAS}{Distributed Antenna System}
\acro{DNA}{DeoxyriboNucleic Acid}

\acro{EC}{European Commission}
\acro{EED}[EED]{exact eigenvalues distribution}
\acro{EIRP}{Equivalent Isotropically Radiated Power}
\acro{ELP}{equivalent low-pass}
\acro{eMBB}{Enhanced Mobile Broadband}
\acro{EMF}{ElectroMagnetic Field}
\acro{EU}{European union}
\acro{EI}{Exposure Index}
\acro{eICIC}{enhanced Inter-Cell Interference Coordination}

\acro{FC}[FC]{fusion center}
\acro{FCC}{Federal Communications Commission}
\acro{FEC}{forward error correction}
\acro{FFT}{fast Fourier transform}
\acro{FH}{frequency-hopping}
\acro{FH-SS}{frequency-hopping spread-spectrum}
\acrodef{FS}{Frame synchronization}
\acro{FSsmall}[FS]{frame synchronization}  
\acro{FDMA}{Frequency Division Multiple Access}

\acro{GA}{Gaussian approximation}
\acro{GF}{Galois field }
\acro{GG}{Generalized-Gaussian}
\acro{GIC}[GIC]{generalized information criterion}
\acro{GLRT}{generalized likelihood ratio test}
\acro{GPS}{Global Positioning System}
\acro{GMSK}{Gaussian Minimum Shift Keying}
\acro{GSMA}{Global System for Mobile communications Association}
\acro{GS}{ground station}
\acro{GMG}{ Grid-connected MicroGeneration}

\acro{HAP}{high altitude platform}
\acro{HetNet}{Heterogeneous network}

\acro{IDR}{information distortion-rate}
\acro{IFFT}{inverse fast Fourier transform}
\acro{iht}[IHT]{iterative hard thresholding}
\acro{i.i.d.}{independent, identically distributed}
\acro{IoT}{Internet of Things}                      
\acro{IR}{impulse radio}
\acro{lric}[LRIC]{lower restricted isometry constant}
\acro{lrict}[LRICt]{lower restricted isometry constant threshold}
\acro{ISI}{intersymbol interference}
\acro{ITU}{International Telecommunication Union}
\acro{ICNIRP}{International Commission on Non-Ionizing Radiation Protection}
\acro{IEEE}{Institute of Electrical and Electronics Engineers}
\acro{ICES}{IEEE international committee on electromagnetic safety}
\acro{IEC}{International Electrotechnical Commission}
\acro{IARC}{International Agency on Research on Cancer}
\acro{IS-95}{Interim Standard 95}

\acro{KPI}{Key Performance Indicator}

\acro{LEO}{low earth orbit}
\acro{LF}{likelihood function}
\acro{LLF}{log-likelihood function}
\acro{LLR}{log-likelihood ratio}
\acro{LLRT}{log-likelihood ratio test}
\acro{LoS}{Line-of-Sight}
\acro{LRT}{likelihood ratio test}
\acro{wlric}[LWRIC]{lower weak restricted isometry constant}
\acro{wlrict}[LWRICt]{LWRIC threshold}
\acro{LPWAN}{Low Power Wide Area Network}
\acro{LoRaWAN}{Low power long Range Wide Area Network}
\acro{NLoS}{Non-Line-of-Sight}
\acro{LiFi}[Li-Fi]{light-fidelity}
 \acro{LED}{light emitting diode}
 \acro{LABS}{LoS transmission with each ABS}
 \acro{NLABS}{NLoS transmission with each ABS}

\acro{MB}{multiband}
\acro{MC}{macro cell}
\acro{MDS}{mixed distributed source}
\acro{MF}{matched filter}
\acro{m.g.f.}{moment generating function}
\acro{MI}{mutual information}
\acro{MIMO}{Multiple-Input Multiple-Output}
\acro{MISO}{multiple-input single-output}
\acrodef{maxs}[MJSO]{maximum joint support cardinality}                       
\acro{ML}[ML]{maximum likelihood}
\acro{MMSE}{minimum mean-square error}
\acro{MMV}{multiple measurement vectors}
\acrodef{MOS}{model order selection}
\acro{M-PSK}[${M}$-PSK]{$M$-ary phase shift keying}                       
\acro{M-APSK}[${M}$-PSK]{$M$-ary asymmetric PSK} 
\acro{MP}{ multi-period}
\acro{MINLP}{mixed integer non-linear programming}

\acro{M-QAM}[$M$-QAM]{$M$-ary quadrature amplitude modulation}
\acro{MRC}{maximal ratio combiner}                  
\acro{maxs}[MSO]{maximum sparsity order}                                      
\acro{M2M}{Machine-to-Machine}                                                
\acro{MUI}{multi-user interference}
\acro{mMTC}{massive Machine Type Communications}      
\acro{mm-Wave}{millimeter-wave}
\acro{MP}{mobile phone}
\acro{MPE}{maximum permissible exposure}
\acro{MAC}{media access control}
\acro{NB}{narrowband}
\acro{NBI}{narrowband interference}
\acro{NLA}{nonlinear sparse approximation}
\acro{NLOS}{Non-Line of Sight}
\acro{NTIA}{National Telecommunications and Information Administration}
\acro{NTP}{National Toxicology Program}
\acro{NHS}{National Health Service}

\acro{LOS}{Line of Sight}

\acro{OC}{optimum combining}                             
\acro{OC}{optimum combining}
\acro{ODE}{operational distortion-energy}
\acro{ODR}{operational distortion-rate}
\acro{OFDM}{Orthogonal Frequency-Division Multiplexing}
\acro{omp}[OMP]{orthogonal matching pursuit}
\acro{OSMP}[OSMP]{orthogonal subspace matching pursuit}
\acro{OQAM}{offset quadrature amplitude modulation}
\acro{OQPSK}{offset QPSK}
\acro{OFDMA}{Orthogonal Frequency-division Multiple Access}
\acro{OPEX}{Operating Expenditures}
\acro{OQPSK/PM}{OQPSK with phase modulation}

\acro{PAM}{pulse amplitude modulation}
\acro{PAR}{peak-to-average ratio}
\acrodef{pdf}[PDF]{probability density function}                      
\acro{PDF}{probability density function}
\acrodef{p.d.f.}[PDF]{probability distribution function}
\acro{PDP}{power dispersion profile}
\acro{PMF}{probability mass function}                             
\acrodef{p.m.f.}[PMF]{probability mass function}
\acro{PN}{pseudo-noise}
\acro{PPM}{pulse position modulation}
\acro{PRake}{Partial Rake}
\acro{PSD}{power spectral density}
\acro{PSEP}{pairwise synchronization error probability}
\acro{PSK}{phase shift keying}
\acro{PD}{power density}
\acro{8-PSK}[$8$-PSK]{$8$-phase shift keying}
\acro{PPP}{Poisson point process}
\acro{PCP}{Poisson cluster process}
 
\acro{FSK}{Frequency Shift Keying}

\acro{QAM}{Quadrature Amplitude Modulation}
\acro{QPSK}{Quadrature Phase Shift Keying}
\acro{OQPSK/PM}{OQPSK with phase modulator }

\acro{RD}[RD]{raw data}
\acro{RDL}{"random data limit"}
\acro{ric}[RIC]{restricted isometry constant}
\acro{rict}[RICt]{restricted isometry constant threshold}
\acro{rip}[RIP]{restricted isometry property}
\acro{ROC}{receiver operating characteristic}
\acro{rq}[RQ]{Raleigh quotient}
\acro{RS}[RS]{Reed-Solomon}
\acro{RSC}[RSSC]{RS based source coding}
\acro{r.v.}{random variable}                               
\acro{R.V.}{random vector}
\acro{RMS}{root mean square}
\acro{RFR}{radiofrequency radiation}
\acro{RIS}{Reconfigurable Intelligent Surface}
\acro{RNA}{RiboNucleic Acid}
\acro{RRM}{Radio Resource Management}
\acro{RUE}{reference user equipments}
\acro{RAT}{radio access technology}
\acro{RB}{resource block}

\acro{SA}[SA-Music]{subspace-augmented MUSIC with OSMP}
\acro{SC}{small cell}
\acro{SCBSES}[SCBSES]{Source Compression Based Syndrome Encoding Scheme}
\acro{SCM}{sample covariance matrix}
\acro{SEP}{symbol error probability}
\acro{SG}[SG]{sparse-land Gaussian model}
\acro{SIMO}{single-input multiple-output}
\acro{SINR}{signal-to-interference plus noise ratio}
\acro{SIR}{signal-to-interference ratio}
\acro{SISO}{Single-Input Single-Output}
\acro{SMV}{single measurement vector}
\acro{SNR}[\textrm{SNR}]{signal-to-noise ratio} 
\acro{sp}[SP]{subspace pursuit}
\acro{SS}{spread spectrum}
\acro{SW}{sync word}
\acro{SAR}{specific absorption rate}
\acro{SSB}{synchronization signal block}
\acro{SR}{shrink and realign}

\acro{tUAV}{tethered Unmanned Aerial Vehicle}
\acro{TBS}{terrestrial base station}

\acro{uUAV}{untethered Unmanned Aerial Vehicle}
\acro{PDF}{probability density functions}

\acro{PL}{path-loss}

\acro{TH}{time-hopping}
\acro{ToA}{time-of-arrival}
\acro{TR}{transmitted-reference}
\acro{TW}{Tracy-Widom}
\acro{TWDT}{TW Distribution Tail}
\acro{TCM}{trellis coded modulation}
\acro{TDD}{Time-Division Duplexing}
\acro{TDMA}{Time Division Multiple Access}
\acro{Tx}{average transmit}

\acro{UAV}{Unmanned Aerial Vehicle}
\acro{uric}[URIC]{upper restricted isometry constant}
\acro{urict}[URICt]{upper restricted isometry constant threshold}
\acro{UWB}{ultrawide band}
\acro{UWBcap}[UWB]{Ultrawide band}   
\acro{URLLC}{Ultra Reliable Low Latency Communications}
         
\acro{wuric}[UWRIC]{upper weak restricted isometry constant}
\acro{wurict}[UWRICt]{UWRIC threshold}                
\acro{UE}{User Equipment}
\acro{UL}{UpLink}

\acro{WiM}[WiM]{weigh-in-motion}
\acro{WLAN}{wireless local area network}
\acro{wm}[WM]{Wishart matrix}                               
\acroplural{wm}[WM]{Wishart matrices}
\acro{WMAN}{wireless metropolitan area network}
\acro{WPAN}{wireless personal area network}
\acro{wric}[WRIC]{weak restricted isometry constant}
\acro{wrict}[WRICt]{weak restricted isometry constant thresholds}
\acro{wrip}[WRIP]{weak restricted isometry property}
\acro{WSN}{wireless sensor network}                        
\acro{WSS}{Wide-Sense Stationary}
\acro{WHO}{World Health Organization}
\acro{Wi-Fi}{Wireless Fidelity}

\acro{sss}[SpaSoSEnc]{sparse source syndrome encoding}

\acro{VLC}{Visible Light Communication}
\acro{VPN}{Virtual Private Network} 
\acro{RF}{Radio Frequency}
\acro{FSO}{Free Space Optics}
\acro{IoST}{Internet of Space Things}

\acro{GSM}{Global System for Mobile Communications}
\acro{2G}{Second-generation cellular network}
\acro{3G}{Third-generation cellular network}
\acro{4G}{Fourth-generation cellular network}
\acro{5G}{Fifth-generation cellular network}	
\acro{gNB}{next-generation Node-B Base Station}
\acro{NR}{New Radio}
\acro{UMTS}{Universal Mobile Telecommunications Service}
\acro{LTE}{Long Term Evolution}

\acro{QoS}{Quality of Service}
\end{acronym}
	
\newcommand{\SAR} {\mathrm{SAR}}
\newcommand{\WBSAR} {\mathrm{SAR}_{\mathsf{WB}}}
\newcommand{\gSAR} {\mathrm{SAR}_{10\si{\gram}}}
\newcommand{\Sab} {S_{\mathsf{ab}}}
\newcommand{\Eavg} {E_{\mathsf{avg}}}
\newcommand{\ft}{f_{\textsf{th}}}
\newcommand{\alphatf}{\alpha_{24}}

\title{Fundamental limits via CRB of semi-blind channel estimation in Massive MIMO systems
}
\author{
Xue Zhang, Abla Kammoun, {\em Member, IEEE}, and Mohamed-Slim Alouini, {\em Fellow, IEEE}
\thanks{The authors are with Computer, Electrical and Mathematical Sciences
and Engineering (CEMSE) Division, Department of Electrical and Computer
Engineering, King Abdullah University of Science and Technology (KAUST), Thuwal 23955-6900, Saudi Arabia. (e-mail: xue.zhang@kaust.edu.sa; abla.kammoun@kaust.edu.sa; slim.alouini@kaust.edu.sa). 
}
\vspace{-8mm}
}
\maketitle

\vspace{-0.8cm}

\begin{abstract}
This paper investigates the asymptotic behavior of the deterministic and stochastic Cram\'er-Rao Bounds (CRB) for semi-blind channel estimation in massive multiple-input multiple-output (MIMO) systems. We derive and analyze mathematically tractable expressions for both metrics under various asymptotic regimes, which govern the growth rates of the number of antennas, the number of users, the training sequence length, and the transmission block length. Unlike the existing work, our results show that the CRB can be made arbitrarily small as the transmission block length increases, but only when the training sequence length grows at the same rate and the number of users remains fixed. However, if the number of training sequences remains proportional to the number of users, the channel estimation error is always lower-bounded by a non-vanishing constant. Numerical results are presented to support our findings and demonstrate the advantages of semi-blind channel estimation in reducing the required number of training sequences.
\end{abstract}

\begin{IEEEkeywords}
CRB, semi-blind, channel estimation, random matrix theory. 
\end{IEEEkeywords}

\section{Introduction}
\label{sec_introdution}
In multiple-input multiple-output (MIMO) systems, a crucial factor in realizing high capacity gains is accurate channel state information (CSI), which relies heavily on the precision of channel estimation. Consequently, channel estimation is vital for obtaining reliable CSI in wireless communications. One approach is data-aided estimation, such as blind channel estimation, which relies solely on the received data signal~\cite{noh2014new}. This method fully exploits the statistical properties of the unknown transmitted symbols, enhancing spectral efficiency and reducing overhead. However, blind channel estimation introduces ambiguities in identifying the channel coefficients and increases computational complexity~\cite{tong1994blind}.
Alternatively, channel estimation can be achieved using training sequences, known as pilot symbols, which are inserted into the data frames as prior information~\cite{hassibi2003much}. This method is widely adopted for its robustness and low computational complexity. The primary drawback, however, is the additional bandwidth consumption required for the pilot symbols.

\par 

To address the aforementioned challenges, semi-blind channel estimation utilizes information from unknown data symbols alongside pilot sequences. This approach enables the estimation of channel coefficients with comparable accuracy while requiring fewer pilot symbols~\cite{srinivas2019iterative}. Unlike blind channel estimation, which relies solely on data symbols, semi-blind channel estimation uses a limited number of pilot symbols with blind methods to resolve ambiguities arising from unknown data symbols and enhance the quality of channel coefficient estimates~\cite{aldana2003channel}. Channel estimation in massive MIMO systems presents significant challenges due to the high number of parameters that need to be estimated. To address this, massive MIMO systems predominantly utilize time-division duplex (TDD) protocols, rather than frequency-division duplex (FDD), to reduce the required overhead~\cite{marzetta2010noncooperative,rusek2012scaling}. TDD systems not only reduce overhead but also provide additional advantages over FDD, such as improved uplink detection and enhanced downlink precoding accuracy, owing to the assumption of reciprocity between the uplink and downlink channels~\cite{nayebi2017semi}. 
 In this context, semi-blind channel estimation, which combines information from both training and data symbols, emerges as a promising technique for further improving channel estimation quality without increasing the training sequence overhead.
 
\par

A substantial body of literature investigates semi-blind channel estimation schemes in multi-user MIMO systems. Among these, maximum-likelihood estimation techniques stand out as prominent methods, particularly expectation-maximization (EM) algorithms. In \cite{aldana2003channel}, the authors introduced a frequency-domain EM algorithm for estimating both blind and semi-blind channels for each user in an underdetermined multiple-input single-output (MISO) system. The study in \cite{nayebi2017semi} examined two EM-based schemes for semi-blind channel estimation under commonly accepted assumptions. Additionally, \cite{al2021semi} utilized eigenvalue decomposition to enhance the computational efficiency of the EM algorithm with a Gaussian prior, while also leveraging the actual discrete prior of the data symbols to derive a tractable version of the EM algorithm.
Iterative techniques can also be employed to compute maximum-likelihood channel estimators. In this context, the authors in \cite{abuthinien2008semi} proposed an iterative two-level optimization loop for jointly estimating channel coefficients and training sequences in MIMO systems.
Besides maximum-likelihood estimators, another class of methods relies on blind techniques, which can be viewed as extensions of blind subspace methods that incorporate information from training sequences to address the matrix ambiguity inherent in blind methods. For instance, in \cite{jagannatham2006whitening}, the MIMO channel matrix was decomposed into the product of two matrices: one being a whitening matrix obtained through blind estimation, and the other a rotational unitary matrix estimated with the assistance of training symbols.
Furthermore, the authors in \cite{rekik2024fast} devised a blind and semi-blind subspace-based algorithm aimed at reducing computational costs, accelerating convergence, and ensuring accurate estimates with a smaller sample size. In \cite{lawal2023semi}, a semi-blind structured signal subspace algorithm was developed to reduce both computational complexity and the training sequence size.


\par

To analytically assess the performance of the aforementioned algorithms, Cram\'er-Rao bounds (CRBs) have been employed for semi-blind channel estimation. We distinguish between two types of CRBs: the deterministic CRB, which assumes that data symbols are deterministic, and the stochastic CRB, which assumes that data is drawn from a specific distribution \cite{stoica1990performance,sandkuhler1987accuracy}. In single-input multiple-output (SIMO) systems, the authors in \cite{de1997cramer} investigated and compared the CRBs for semi-blind, blind, and pilot symbol-based channel estimation. Additionally, the asymptotic performance of semi-blind estimators was examined in \cite{de1997cramer,de1997asymptotic}, focusing on SIMO systems when the lengths of the training sequences and data symbols approach infinity.

In the context of massive MIMO systems, more recent work by \cite{nayebi2017semi} studied the asymptotic behavior of CRBs with an unlimited number of antennas. The results indicate that as the number of antennas increases, the deterministic CRB converges to that of a system where all data symbols are known, while the stochastic CRB approaches the CRB of a system utilizing orthogonal training sequences that span the entire transmission block. However, these findings are not comprehensive. A notable limitation is that this study only considers the scenario where the number of antennas grows infinitely large, neglecting other cases where the transmission block length and the number of users scale proportionally with the number of antennas. This omission is significant, especially given that many practical massive MIMO systems operate in regimes where the channel remains constant over hundreds or even tens of thousands of samples~\cite{bjornson2017massive}. In such cases, the CRBs expressions provided in \cite{nayebi2017semi} are high-dimensional and analytically intractable, thus simplifying these expressions and
establishing their connections to random matrix theory remains a significant challenge.

To address this gap, we propose a more comprehensive approach for analyzing deterministic and stochastic CRBs using random matrix theory. Overall, our contributions are summarized as follows:
\begin{itemize}
    \item We derive mathematically tractable formulas for the normalized traces of the deterministic and stochastic CRBs, referred to as averaged CRBs. Specifically, we show that the average deterministic CRB depends exclusively on the data matrix, which is assumed to be deterministic, whereas the average stochastic CRB depends on the channel matrix. Our derived formulas are exact and are expressed in a significantly simpler form than those previously obtained in \cite{nayebi2017semi}. 
    \item We analyze the asymptotic behaviors of both the average deterministic and stochastic CRBs by using Stieltjes transform across various regimes related to the number of users, the number of antennas, the number of pilot symbols, and the total length of the data transmission block. Our findings indicate that the average deterministic and stochastic CRBs do not decrease indefinitely as the total block length increases unless the number of training symbols grows at a comparable rate and the number of users remains fixed. This suggests that, while incorporating data transmission in semi-blind channel estimation is advantageous, it does not lead to further reductions in channel estimation error without incurring the additional cost of more training symbols. Furthermore, our analysis reveals a power-dependent behavior: when the training symbols are allocated more power than the unknown data symbols, increasing the number of training symbols enhances channel estimation performance. Conversely, when the training symbols have lower power, it becomes more effective to reduce the number of training symbols and rely more on the data symbols for semi-blind channel estimation.
    \item We use numerical results to validate the accuracy of the proposed asymptotic expressions for the average deterministic and stochastic CRBs and demonstrate practical applications of these obtained CRBs.
\end{itemize}

\textit{Organization}: The remainder of this paper is organized as follows. Section~\ref{sec_model} introduces the system model. In Section~\ref{sec:deterministic}, we derive the closed-form expression of the average deterministic CRB and analyze its asymptotic behavior in different asymptotic regimes. In Section~\ref{sec_SCRB}, we derive mathematically tractable formulas for the average stochastic CRB and analyze its asymptotic behavior in different asymptotic regimes. Simulation results are provided in Section~\ref{sec:sim}. Finally, Section~\ref{sec_con} concludes the article.

\textit{Notations}: The vectors (matrices) are denoted by lower-case (upper-case) boldface characters. $\mathbf{I}_{M}$ denotes the $M\times M$ identity matrix, and $\mathbf{0}$ stands for a zero matrix with appropriate dimension. $\mathbb{E}\{\cdot\}$ represents the expectation operator. The trace of $\mathbf{A}$ are denoted by $\mathrm{tr}(\mathbf{A})$. Superscripts $T$, $H$, and $\ast$ denote the transpose, conjugate transpose, and complex conjugate, respectively. $[\mathbf{A}]_{i,j}$ stands for the $(i,j)$ entry of $\mathbf{A}$, especially, $[\mathbf{A}]_{i}$ represents the $(i,i)$ entry of $\mathbf{A}$. $\mathrm{diag}(\mathbf{a})$ represents a diagonal matrix that employs the elements of $\mathbf{a}$ as its diagonal entries. $|\cdot|$ denotes the absolute value. $\otimes$ represents the Kronecker product. $\|\cdot\|_2$  stand for the $l_2$ norm. $\mathbf{A}\succeq\mathbf{0}$ is equivalent to $\mathbf{A}$ being positive semidefinite. $\mathrm{det}(\mathbf{A})$ denotes the determinant of $\mathbf{A}$. $\overset{a.s.}{\longrightarrow}$ represents almost sure (a.s.) convergence. $\delta(x)$ denotes the Dirac measure at point $x$. $\mathrm{Re}\{\cdot\}$ and $\mathrm{Im}\{\cdot\}$ represent real and imaginary parts, respectively. 
\section{System Model}
\label{sec_model}
We consider a single-cell system composed of a base station (BS) equipped with $M$ antennas to serve $K$ single-antenna users where $K<M$. The channel matrix between the BS and users is expressed as $\mathbf{G}=\mathbf{HB}^{1/2}$, where $\mathbf{H}=[\mathbf{h}_1,\ldots,\mathbf{h}_K]$ is a matrix representing small scale fading, and $\mathbf{B}=\mathrm{diag}([\rho_1,\ldots,\rho_k])$ is a diagonal matrix with $\rho_k$ $(k=1,\ldots, K)$ being the large scale fading coefficient between the BS and the $k$th user. We consider a time block fading model, where columns of $\mathbf{H}$ are constant during a block of $N$ symbols, and change to independent values at the next coherence block.  

Moreover, we consider uplink transmission where users send $L$ known training sequences followed by $(N-L)$ unknown data symbols. The uplink signal received by BS at time $n$ is given by 
\begin{align}\label{received signal}
    \mathbf{y}(n) = \mathbf{G}\mathbf{s}(n) + \mathbf{v}(n),
\end{align}
where $\mathbf{s}(n)=[s_1(n),\ldots,s_K(n)]^T$ for $n=0,\ldots,L-1$ are known training sequences, $\mathbf{s}(n)=[s_1(n),\ldots,s_K(n)]^T$ for $n=L,\ldots,N-1$ are the unknown data symbols with unit power $\mathbb{E}\left\{\mathbf{s}(n)\mathbf{s}(n)^H\right\}=\mathbf{I}_K$, and $\mathbf{v}(n)\sim\mathcal{CN}(\pmb{0},\sigma_v^2\mathbf{I}_M)$ is additive Gaussian noise vector. Let $\mathbf{S}_p=[\mathbf{s}(0),\ldots,\mathbf{s}(L-1)]$ and $\mathbf{S}_d=[\mathbf{s}(L),\ldots,\mathbf{s}(N-1)]$ represent the known training sequences and unknown data symbols during the channel coherence period, respectively. Hence, the complete transmit symbols are given by $\mathbf{S}=[\mathbf{S}_p,\mathbf{S}_d]$. Correspondingly, the received signal and noise can be given by $\mathbf{Y}=[\mathbf{Y}_p,\mathbf{Y}_d]$ and $\mathbf{V}=[\mathbf{V}_p,\mathbf{V}_d]$, respectively. 

In this work, we adopt a semi-blind channel estimation framework, which leverages both the known training sequences $\mathbf{S}_p$ and the statistical structure of the unknown data symbols $\mathbf{S}_d$ to estimate the channel $\mathbf{G}$ in (\ref{received signal}). Within this framework, we derive mathematically tractable expressions for the CRB for semi-blind channel estimation under two key assumptions. The first assumption models the data signal matrix ${\bf S}_d$ as an unknown deterministic quantity, resulting in the deterministic CRB. The second assumption considers ${\bf S}_d$ as a Gaussian matrix. While a similar study has been carried out in \cite{nayebi2017semi}, the obtained CRB expressions are not readily amenable to analysis in the random matrix regime, due to their high dimensionality, mathematical complexity, and the lack of a direct connection to well-established random objects, in particular the Stieltjes transform of empirical measures of some random matrices. In this work, we show that by simplifying the derivations, we can obtain more practical and tractable formulas that are suitable for study in the random matrix regime.

\par

\section{Analysis of deterministic Cram\'{e}r Rao bound for semi-blind channel estimation}
\label{sec:deterministic}
\subsection{Derivation of deterministic Cram\'{e}r Rao bound}
In computing the deterministic CRB, the data matrix ${\bf S}_d$ is treated as an unknown deterministic matrix. Under this model, we start by modeling the received data symbols as independent complex Gaussian signals with the following parameters:
$$
{\bf y}(n)\sim \mathcal{C}\mathcal{N}({\bf G}{\bf s}(n), \sigma_v^2\mathbf{I}_M), \ \ n=0,\cdots, N-1,
$$
where ${\bf s}(n)$ refers to a training symbol if $n=0,\cdots,L-1$ or ${\bf s}(n)$ refers to an unknown data symbols otherwise. The following theorem provides expressions for the deterministic CRB. Although the real CRB expression has been derived in \cite{nayebi2017semi}, its high-dimensional matrix form makes it difficult to analyze directly using random matrix theory. To address this, we rederive the complex CRB with Kronecker product, which yields an equivalent lower bound to that in \cite{nayebi2017semi} but in a more tractable form. For completeness, we restate the result and provide a detailed proof.
\begin{theorem}[\cite{nayebi2017semi}]
\label{th:det_crb}
Define matrices $\boldsymbol{\Lambda}$ and $(\boldsymbol{\Omega}_n)_{n=L}^{N-1}$ as: 
\begin{align}\label{lambda,omega}
    &\mathbf{\Lambda} = \frac{1}{\sigma_v^2}\sum_{n=0}^{N-1}\mathbf{s}^{\ast}(n)\mathbf{s}(n)^T\otimes\mathbf{I}_M, \\
    &\mathbf{\Omega}_n = \frac{1}{\sigma_v^2}\mathbf{s}(n)^T\otimes\mathbf{G}^H, \,\, n= L,\ldots,N-1.
\end{align}
Then, the deterministic CRB of the  channel matrix is given by:
\begin{align}\label{DCRB_expression}
       \mathbf{CRB} = \left(\mathbf{\Lambda}-\sum_{n=L}^{N-1}\mathbf{\Omega}_n^H\mathbf{A}^{-1}\mathbf{\Omega}_n\right)^{-1},
\end{align}
where $\mathbf{A}=\frac{1}{\sigma_v^2}\mathbf{G}^H\mathbf{G}$.
\end{theorem}
\begin{proof}
See Appendix~\ref{derivation_DCRB}.
\end{proof}
Recalling that the goal of this paper is to study the CRB based on the statistical properties of the channel and data symbols, we note that the expression derived in Theorem \ref{th:det_crb} poses mathematical challenges due to the intricate dependency of the CRB on these quantities. 
In the following corollary, we demonstrate how certain simplifications can be applied to make the expression more tractable, facilitating asymptotic analysis under various regimes.
\begin{corollary}
Define ${\bf X}={\bf S}{\bf S}^{H}$, ${\bf X}_p={\bf S}_p{\bf S}_p^{H}$ and ${\bf P}={\bf G}({\bf G}^{H}{\bf G})^{-1}{\bf G}^{H}$, the projection matrix on the space spanned by the columns of ${\bf G}$. Assume that ${M}\geq K$, $K\leq L$, ${\bf G}$ is a full-rank matrix and ${\bf X}_p$ is non-singular. Then, the ${\bf CRB}$ defined in \eqref{DCRB_expression} can be simplified to:
\begin{align}
{\bf CRB}=\sigma_v^2((\mathbf{X}^{\ast})^{-1}\otimes\mathbf{P}^{\bot}+(\mathbf{X}_p^{\ast})^{-1}\otimes\mathbf{P}). 
\end{align}
where $\mathbf{P}^{\bot}=\mathbf{I}_M-{\bf P}.$
\label{cor:crb}
\end{corollary}
\begin{proof}
Using Kronecker product property $({\bf A}\otimes {\bf B})({\bf C}\otimes {\bf D})={\bf AC}\otimes {\bf BD}$ which holds for any matrices ${\bf A}$, $ {\bf B}$, ${\bf C}$ and $ {\bf D}$ with sizes allowing to form the matrix products ${\bf AC}$ and ${\bf BD}$, we can write $\boldsymbol{\Omega}_n^{H}{\bf A}^{-1}\boldsymbol{\Omega}_n$ as $\frac{1}{\sigma_v^2}{\bf s}(n)^\ast{\bf s}(n)^{T}\otimes {\bf P}$. To continue, we extend the definition of matrices $\boldsymbol{\Omega}_n$ to $n=0,\cdots, L-1$ by defining $\boldsymbol{\Omega}_n=\frac{1}{\sigma_v^2}{\bf s}(n)^{T}\otimes {\bf G}^{H}$. With this, we may simplify the CRB as:
\begin{align}
    \mathbf{CRB} =& \Big(\frac{1}{\sigma_v^2}\sum_{n=0}^{N-1}({\bf s}(n)^{\ast}{\bf s}(n)^{T}\otimes \mathbf{I}_M-{\bf s}(n)^{\ast}{\bf s}(n)^{T}\otimes{\bf P}) \notag\\
    &+\frac{1}{\sigma_v^2}\sum_{n=0}^{L-1}{\bf s}(n)^{\ast}{\bf s}(n)^{T}\otimes{\bf P}\Big)^{-1} \notag \\
    =& \sigma_v^2\left(\mathbf{S}^{\ast}\mathbf{S}^T\otimes(\mathbf{I}_M-\mathbf{P}) + \mathbf{S}_p^{\ast}\mathbf{S}_p^T\otimes\mathbf{P}\right)^{-1} \notag \\
    =& \sigma_v^2((\mathbf{X}^{\ast})\otimes\mathbf{P}^{\bot}+(\mathbf{X}_p^{\ast})\otimes\mathbf{P})^{-1} \notag \\
    =&\sigma_v^2((\mathbf{X}^{\ast})^{-1}\otimes\mathbf{P}^{\bot}+(\mathbf{X}_p^{\ast})^{-1}\otimes\mathbf{P}),
\end{align}
where the last equality follows by checking that:
$$
((\mathbf{X}^{\ast})\otimes\mathbf{P}^{\bot}+(\mathbf{X}_p^{\ast})\otimes\mathbf{P})((\mathbf{X}^{\ast})^{-1}\otimes\mathbf{P}^{\bot}+(\mathbf{X}_p^{\ast})^{-1}\otimes\mathbf{P})=\mathbf{I}_M.
$$
\end{proof}
\noindent{\bf On the case $L\leq K$}
It is worth noting that for the CRB to exist, the matrix ${\bf X}_p$ must be non-singular. If ${\bf X}_p$ is singular, the Fisher-information matrix becomes singular as well, making the CRB undefined. A necessary condition for ${\bf X}_p$ to be non-singular is that $K\leq L$. This requirement is the same as that imposed by training-based schemes, where the number of training symbols must be at least equal to the number of users. However, in semi-blind channel estimation, one may think that there is potential to reduce the required number of training symbols by initially applying a blind channel estimation technique in combination with a training-based technique. From the CRB calculations, such a procedure could not reduce the number of required training symbols to less than $K$. To investigate further this scenario, consider a subspace blind channel estimate. It is well-known in this case in the noiseless case, the channel is estimated up to a scalar ambiguity~\cite{gao2007blind,shin2007blind,ghavami2017blind}; that is there exists a matrix ${\bf F}$ of size $K\times K$ such that the blind channel estimate $\hat{\bf G}_{\mathrm{blind}}={\bf G}{\bf F}$. In the following Lemma, we show that to resolve the matrix ambiguity, at least $K$  training symbols are required. 
\begin{lemma}\label{lemma1}
    For solving the matrix ambiguity in a semi-blind subspace technique, the number of training symbols $L$ should be greater than $K$.
\begin{proof}
    Consider the subspace blind channel scenario, the estimate of $\mathbf{G}$ can be obtained within an ambiguity matrix, namely, $\hat{\mathbf{G}}_{\mathrm{blind}}=\mathbf{GF}$, where $\mathbf{F}\in\mathbb{C}^{K\times K}$ is an unknown invertible matrix. Therefore, in the noiseless case, the received signal can be rewritten as $\mathbf{y}(n)=\hat{\mathbf{G}}_{\mathrm{blind}}\mathbf{F}^{-1}\mathbf{s}(n)$. Performing singular value decomposition (SVD) of $\hat{\mathbf{G}}_{\mathrm{blind}}$, we have $\hat{\mathbf{G}}_{\mathrm{blind}}=\mathbf{U}\mathbf{D}\mathbf{V}^H$, where $\mathbf{D}\in\mathbb{C}^{K\times K}$ is a diagonal matrix whose diagonal entries correspond to the singular values of $\hat{\mathbf{G}}$, $\mathbf{U}\in\mathbb{C}^{M\times K}$ and $\mathbf{V}\in\mathbb{C}^{K\times K}$ are the left and right singular vector matrices, respectively. For each \(n = 0, \dots, L-1\), the equation \(\mathbf{U}^H\mathbf{y}(n) = \mathbf{D}\mathbf{V}^H\mathbf{F}^{-1}\mathbf{s}(n)\) provides at most \(K\) independent linear equations, since \(\mathbf{U}^H\) is a \(K \times M\) matrix. Given that \(\mathbf{F}^{-1}\) contains \(K^2\) unknowns, it is necessary to have at least \(L = K\) observations in order to form \(K^2\) independent equations, thereby matching the number of unknowns in \(\mathbf{F}\).
\end{proof}
\end{lemma}
\noindent{\bf Average CRB. } We define the average CRB as:
$$
{\rm CRB}^{\rm avg}:=\frac{1}{K}{\rm tr}({\bf CRB}).
$$
It represents a lower-bound on the per-user mean square error (MSE) of any unbiased estimator of the channel.
More specifically, let $\hat{\bf G}$ be an unbiased estimation of ${\bf G}$. Then, the average deterministic CRB provides a lower-bound on the MSE achieved by $\hat{\bf G}$, in the sense that $\frac{1}{K}\mathbb{E}[\|\hat{\bf G}-{\bf G}\|_{F}^2]\geq {\rm CRB}^{\rm avg}$, where the expectation is taken over the distribution of the noise. 
Using the fact that ${\rm tr}({\bf P})=K$ and ${\rm tr}({\bf P}^{\perp})=M-K$, the average CRB simplifies to:
\begin{align}\label{eq:deterministic_crb}
{\rm CRB}^{\rm avg}=\frac{(M-K)\sigma_v^2}{KN}{\rm tr}((\tilde{\mathbf{X}}^\ast)^{-1})+\frac{\sigma_v^2}{L}{\rm tr}((\tilde{\bf X}_p^\ast)^{-1}),
\end{align}
where $\tilde{\mathbf{X}}=\frac{1}{N}{\bf X}$ and $\tilde{\mathbf{X}}_p=\frac{1}{L}{\bf X}_p$
It is worth noting that the average CRB does not depend on the channel but only on the training and data matrices ${\bf S}_p$ and ${\bf S}_d$. This particularly indicates that as long as the channel is full-rank, the channel estimation error depends mainly on the number of data and training symbols used to perform channel estimation, as well as the matrix ${\bf S}_p$ containing training symbols.

As will be shown next, the average CRB lends itself to theoretical analysis by leveraging the statistical properties of the data matrix. 

\noindent{\bf Comparison with training-based schemes.} In schemes using only training data, the symbol vectors received during data transmissions are not used for channel estimation purposes. As a result, the CRB reduced to ${\bf CRB}_{\rm training}=\sigma_v^2({\bf X}_p^{\ast})^{-1}\otimes \mathbf{I}_M$, which serves as a lower bound on the mean squared error of training-based channel estimation methods, such as the ML estimator given by $\hat{\mathbf{G}}_{\mathrm{training}}=\mathbf{Y}_p\mathbf{S}_p^H(\mathbf{S}_p\mathbf{S}_p^H)^{-1}$. Accordingly, the average CRB writes as:
\begin{align}
    {\rm CRB}_{\rm training}^{{\rm avg}}=\sigma_v^2\frac{M}{KL}{\rm tr} (\tilde{\bf X}_p^\ast)^{-1}.
\end{align}
To compare with the CRB for semi-blind channel estimation, we decompose ${\rm  CRB}^{\rm avg}$ as:
\begin{align}
{\rm CRB}^{\rm avg}=&\frac{\sigma_v^2M}{LK}{\rm tr}((\tilde{\bf X}_p^\ast)^{-1})+ \frac{(M-K)\sigma_v^2}{KN}{\rm tr}((\tilde{\mathbf{X}}^{\ast}))^{-1} \notag\\
&-\frac{(M-K)\sigma_v^2}{KL}{\rm tr}(\tilde{\mathbf{X}}_p^{\ast})^{-1} \notag\\
=&{\rm CRB}_{\rm training}-{\rm CRB_{\rm gain}},
\end{align}
where ${\rm CRB}_{\rm gain}$ represents the gain obtained from using semi-blind channel estimation and is given by:
\begin{equation}
{\rm CRB}_{\rm gain}=\frac{(M-K)\sigma_v^2}{KL}{\rm tr}(\tilde{\mathbf{X}}_p^{\ast})^{-1}-
\frac{(M-K)\sigma_v^2}{KN}{\rm tr}(\tilde{\mathbf{X}}^{\ast})^{-1}.
\label{eq:average_CRB}\end{equation}
One can easily check that ${\rm CRB}_{\rm gain}$ is positive since $L\leq N$ and $\tilde{\bf X}^\ast\succeq \tilde{\bf X}_p^\ast$. Moreover, the gain in terms of CRB becomes all the more important as $M-K$ is large.

\par

\subsection{Asymptotic analysis of the deterministic CRB}
In this section, we consider analyzing the average deterministic CRB under two distinct growth regimes by establishing a direct connection to the Stieltjes transform of the empirical measure of $\tilde{\mathbf{X}}$, a fundamental tool in random matrix theory. The first growth regime is described in the following assumptions and corresponds to the situation in which the number of training symbols and data symbols, as well as the number of users and transmit antennas grow large with the same pace.  

\begin{assumption}\label{assumption_dcrb_1}
    We assume that the number of users $K$, the number of antennas $M$, the number of pilot symbols $L$, and the transmission block length $N$ go to infinity at the same rate, that is $K\rightarrow+\infty$, $M\rightarrow+\infty$, $L\rightarrow+\infty$, and $N\rightarrow+\infty$ with $K/M\rightarrow c$, $M/N\rightarrow\alpha$, and $L/N\rightarrow\beta$ for some constant $0<c<1$, $0<\alpha<\infty$, and $0<\beta<1$.
\end{assumption}

\begin{assumption}
We assume that the unknown data symbol vectors ${\bf s}_d(n), n=L+1,\cdots,N$ are independent and have independent entries with covariance $P_s \mathbf{I}_K$.
\label{ass:data_model}
\end{assumption}
For a square $K\times K$  matrix ${\bf A}$, we denote by $\mu_{{\bf A}}$ its empirical distribution defined as:
$$
\mu_{{\bf A}}(dx)=\frac{1}{K}\sum_{i=1}^K \delta_{\lambda_i({\bf A})}(dx),
$$
where $\lambda_1,\cdots,\lambda_K$ denote the eigenvalues of ${\bf A}$ and $\delta_x$ is the dirac measure. Associated with the empirical measure, we define its Stieltjes transform as:
$$
m_{{\bf A}}(z)=\int\frac{1}{\lambda-z}\mu_{{\bf A}}(d\lambda)=\frac{1}{K}{\rm tr}({\bf A}-z\mathbf{I}_K)^{-1},
$$
With this definition at hand, we can easily see from \eqref{eq:average_CRB} that the average CRB is a function of the Stietljes transform of matrix $\tilde{\bf X}^{\ast}$. Noting that $\tilde{\bf X}^\ast$ can be written as:
$$
\tilde{\bf X}=\frac{L}{N}\tilde{\bf X}_p+\frac{{\bf S}_d}{\sqrt{K}}\frac{K}{N}\mathbf{I}_{N-L} \frac{{\bf S}_d^{H}}{\sqrt{K}},
$$
we may use the result of \cite{silverstein1995empirical} to show that under the asymptotic regime specified in Assumption \ref{assumption_dcrb_1} and Assumption \ref{ass:data_model}, the Stieltjes transform of the empirical measure of $\tilde{\bf X}(z)$ converges almost surely to:
$$
m_{\tilde{\bf X}}(z)-\overline{m}_{\tilde{\bf X}}(z)\xrightarrow[\substack{K,L, N\to\infty,\frac{K}{N}\to c\alpha,\frac{L}{N}\to \beta} ]{ a.s.} 0,
$$
where  $\overline{m}_{\tilde{\bf X}}(z)$ is the unique solution to the following equation \cite{silverstein1995empirical}:
$$
\overline{m}_{\tilde{\bf X}}(z)=\frac{1}{K}\sum_{k=1}^K \frac{1}{\beta\lambda_k(\tilde{\bf X}_p)-z+\frac{(1-\beta)P_s}{1+c\alpha P_s\overline{m}_{\tilde{\bf X}}(z)}}.
$$
Plugging the above expression into the expression of the average CRB, we thus prove the following theorem:
\begin{theorem}
\label{th:average_crb}
Under Assumption \ref{assumption_dcrb_1} and \ref{ass:data_model}, the average CRB satisfies the following convergence result:
$$
{\rm CRB}^{\rm avg}-\overline{\rm CRB}^{\rm avg}\xrightarrow[ ]{ a.s.} 0,
$$
where
\begin{align}
\overline{\rm CRB}^{\rm avg}=(1-c)\alpha{\sigma_v^2}\overline{m}_{\tilde{\bf X}}(0)+\frac{\sigma_v^2}{L}{\rm tr}((\tilde{\bf X}_p^\ast)^{-1}),
\end{align}
with $\overline{m}_{\tilde{\bf X}}(0)$ being the unique solution to the following fixed-point equation:
\begin{equation}
\overline{m}_{\tilde{\bf X}}(0)=\frac{1}{K}\sum_{k=1}^K \frac{1}{\beta\lambda_k(\tilde{\bf X}_p)+\frac{(1-\beta)P_s}{1+c\alpha P_s \overline{m}_{\tilde{\bf X}}(0)}}. \label{eq:msolution}
\end{equation}
\end{theorem}

\noindent{\bf On the numerical evaluation of $\overline{m}_{\tilde{\bf X}}(0)$.} The solution of \eqref{eq:msolution} in $\overline{m}_{\tilde{\bf X}}(0)$ does not in general have explicit formulation except some particular cases like the one in which all eigenvalues of $\lambda_k(\tilde{\bf X}_p), k=1,\cdots, K$ are equal. However, numerical evaluation of the solution $\overline{m}_{\tilde{\bf X}}(0)$ can be obtained through iterations using the algorithm in \cite{silverstein1995empirical}, which we provide in Algorithm \ref{alg:1} for completeness. To ensure the convergence of this algorithm to the positive solution $\overline{m}_{\tilde{\bf X}}(0)$, we define a function $T(x)$ as
\begin{align}
    T(x) = \frac{1}{K}\sum_{k=1}^K \frac{1}{\beta\lambda_k(\tilde{\bf X}_p)+\frac{(1-\beta)P_s}{1+c\alpha P_s x}}.  \notag 
\end{align}
Since for all $x\geq 0$, the following properties are satisfied: 
\begin{align}
    \left\{\begin{array}{ll}
        T(x) \geq 0,   \\
        T(x) \geq  T(x'),  \textrm{when } x<x',\\
        T(\vartheta x)<\vartheta T(x), \textrm{for all } \vartheta>1.
    \end{array}\right.
\end{align}
 $T(x)$ is a standard interference function~\cite{feyzmahdavian2012contractive}. It can be easily check that the equation $1=\frac{T(x)}{x}$ has a unique positive solution, Hence, using Theorem in \cite{yates1995framework}, the iterations $x^{(l+1)}=T(x^{(l)})$ converges to the unique fixed point, where  $l$ denotes the iteration index. Therefore, the numerical evaluation of $\overline{m}_{\tilde{\bf X}}(0)$ via Algorithm~\ref{alg:1} is guaranteed to converge.
\begin{algorithm}[t]
\caption{Fixed point algorithm.} 
\label{fixed point}
\begin{algorithmic}[1]
\REQUIRE $\tilde{\mathbf{X}}_p$ and $\epsilon>0$. 
\STATE Initialize $m_{\tilde{\mathbf{X}}}^{(0)}(0)=0$ and $l=0$.
\STATE Compute $m_{\tilde{\mathbf{X}}}^{(1)}(0)=\frac{1}{K}\sum\limits_{k=1}^K\frac{1}{\beta\lambda_k(\tilde{\mathbf{X}})+(1-\beta)}$ and set $l=1$.
\WHILE{$|m_{\tilde{\mathbf{X}}}^{(l)}(0)-m_{\tilde{\mathbf{X}}}^{(l-1)}(0)|>\epsilon$}
\STATE Compute $m_{\tilde{\mathbf{X}}}^{(l+1)}(0)=\frac{1}{K}\sum\limits_{k=1}^K\frac{1}{\beta\lambda_k(\tilde{\mathbf{X}})+\frac{(1-\beta)P_s}{1+c\alpha P_sm_{\tilde{\bf X}}^{(l)}(0)}}$. 
\STATE Set $l=l+1$. 
\ENDWHILE
\ENSURE $m_{\tilde{\mathbf{X}}}^{\ast}(0)$.
\end{algorithmic}
\label{alg:1}
\end{algorithm}

\noindent{\bf Optimization of the training symbol matrix.} The tractability of the approximation provided by Theorem \ref{th:average_crb}
 offers an opportunity to refine the design of training symbol sequences to minimize the average CRB. Specifically, we aim to solve the following problem:
\begin{align}\label{optimization problem1}
&\min_{\mathbf{S}_p} \ \   \overline{\rm CRB}^{\rm avg} \notag \\
&\mathrm{s.t.} \,\,\,\, \frac{1}{K}\mathrm{tr}(\tilde{\mathbf{X}}_p) = P,
\end{align}
Using the expression of the average CRB, the above optimization problem can be equivalently rewritten as 
\begin{align}\label{op2}
&\min_{\mathbf{S}_p}\quad (1-c)\alpha{\sigma_v^2}\overline{m}_{\tilde{\bf X}}(0)+\frac{\sigma_v^2}{L}{\rm tr}((\tilde{\bf X}_p^\ast)^{-1}) \notag \\
&\mathrm{s.t.}\quad\,\,\, \overline{m}_{\tilde{\bf X}}(0)=\frac{1}{K}\sum_{k=1}^K \frac{1}{\beta\lambda_k(\tilde{\bf X}_p)+\frac{(1-\beta)P_s}{1+c\alpha P_s\overline{m}_{\tilde{\bf X}}(0)}}, \notag \\
&\qquad\,\,\,\,\, \frac{1}{K}\mathrm{tr}\left(\tilde{\mathbf{X}}_p\right) = P.
\end{align}
The optimal design of matrix ${\bf S}_p$ solving \eqref{op2} is determined in the following theorem:
\begin{theorem}
The optimal matrix $\tilde{\bf X}_p$ that solves \eqref{optimization problem1} is $\tilde{\bf X}_p=P\mathbf{I}_K$.
\end{theorem}
\begin{proof}
By invoking schur-convexity of the function $(x_1,\cdots,x_K)\mapsto \sum_{k=1}^{K}\frac{1}{x_k}$, it can be readily seen that matrix $\tilde{\mathbf{X}}_p=P\mathbf{I}_K$ is the one that minimizes ${\rm tr}((\tilde{\bf X}_p)^{-1})$ under the constraint of non-negative matrices satisfying the constraint $\frac{1}{K}\mathrm{tr} (\tilde{\bf X}_p)=\mathbf{I}_K$. To prove the desired result, it suffices thus to check that $\overline{m}_{\tilde{\bf X}}(0)$ is also minimized when $\tilde{\bf X}_p=P\mathbf{I}_K$. To see this, and in order to emphasize the dependence of $\overline{m}_{\tilde{\bf X}}(0)$, denote by $\overline{m}_{\tilde{\bf X}}(0; \tilde{\bf X}_p)$ the unique solution to the following equation:
$$
\overline{m}_{\tilde{\bf X}}(0;\tilde{\bf X}_p)=\frac{1}{K}\sum_{k=1}^K \frac{1}{\beta\lambda_k(\tilde{\bf X}_p)+\frac{(1-\beta)P_s}{1+c\alpha P_s\overline{m}_{\tilde{\bf X}}(0;\tilde{\bf X}_p)}},
$$
Denote by $f:x\mapsto \frac{1}{\beta x+\frac{1-\beta}{1+c\alpha P_s\overline{m}_{\tilde{\bf X}}(0;\tilde{\bf X}_p)}}$. Using the convexity of function $f$, we thus obtain $
\overline{m}_{\tilde{\bf X}}(0,\tilde{\bf X}_p)\geq \frac{1}{\beta P+\frac{(1-\beta)P_s}{1+c\alpha P_s\overline{m}_{\tilde{\bf X}}(0;\tilde{\bf X}_p)}}$, or equivalently: $1\geq\frac{1}{\beta P\overline{m}_{\tilde{\bf X}}(0;\tilde{\bf X}_p)+\frac{(1-\beta)P_s\overline{m}_{\tilde{\bf X}}(0;\tilde{\bf X}_p)}{1+c\alpha P_s \overline{m}_{\tilde{\bf X}}(0;\tilde{\bf X}_p)}}$.
Recall that $\overline{m}_{\tilde{\bf X}}(0;P\mathbf{I})$ satisfies the following equation:
\begin{equation}
1= \frac{1}{\beta P\overline{m}_{\tilde{\bf X}}(0;P\mathbf{I}_K)+\frac{(1-\beta)P_s\overline{m}_{\tilde{\bf X}}(0;P\mathbf{I}_K)}{1+c\alpha P_s\overline{m}_{\tilde{\bf X}}(0;P\mathbf{I}_K)}}. \notag \label{eq:m_tilde}
\end{equation}
Let ${g}:x\mapsto \frac{1}{\beta Px+\frac{(1-\beta)P_sx}{1+c\alpha P_sx}}$. Clearly, $g$ is non-increasing in the positive axis, we thus deduce that for all $\tilde{\bf X}_p$ positive matrix, with $\frac{1}{K}{\rm tr}(\tilde{\bf X}_p)=1$, $\overline{m}_{\tilde{\bf X}}(0;\tilde{\bf X}_p)\geq \overline{m}_{\tilde{\bf X}}(0;P\mathbf{I}_K)$. 
\end{proof}
From now on, unless otherwise specified, we will work under the assumption that ${\bf X}_p=PL\mathbf{I}_K$ and focus on deriving insights from the asymptotic expression of the CRB.

\begin{figure*}[htbp]
\begin{equation}
\overline{m}_{\tilde{\bf X}}(0)\!=\!\frac{\sqrt{(\beta P-c\alpha P_s+P_s(1-\beta))^2+4c\alpha P_s\beta P}-(\beta P-c\alpha P_s+(1-\beta)P_s)}{2c\alpha\beta{P}P_s}. \label{eq:m_tilde_exact}
\end{equation}
\hrule
\end{figure*}

\vspace{0.2cm}
\noindent{\bf Simplification of  $\overline{\rm{CRB}}^{\rm avg}$ when $\tilde{\bf X}_p$ is optimized. } Replacing $\tilde{\bf X}_p$ by $P\mathbf{I}_K$ into the expression of $\overline{m}_{\tilde{\bf X}}(0)$, we obtain that $\overline{m}_{\tilde{\bf X}}(0)$ is given by \eqref{eq:m_tilde_exact}. The average ${\rm CRB}$ simplifies thus to:
\begin{equation}
\overline{\rm CRB}^{\rm avg}=(1-c)\alpha \sigma_v^2 \overline{m}_{\tilde{\bf X}}(0)+\sigma_v^2\frac{c\alpha}{\beta P}. \label{eq:avg_crb}
\end{equation}

Using \eqref{eq:m_tilde_exact}, the above expression in \eqref{eq:avg_crb} can help serve to determine the $\beta$ for which $\overline{\rm CRB}^{\rm avg}$ is set to a given desired level ${\rm MSE}$. Indeed, assume that the desired level ${\rm MSE}$ is attainable then, the $\beta$ that allows to achieve this level solves the following equation:
\begin{align}
    a\beta^2+b\beta+d=0,    
\end{align}
where 
\begin{align}
a=&\frac{{\rm MSE}P}{cP_s}({\rm MSE}cPP_s+(1-c)\sigma_v^2(P-P_s)), \notag\\
b=&-\frac{\sigma_v^2}{cP_s}((1-c)(\sigma_v^2\alpha(P-cP_s)-{\rm MSE}PP_s)\nonumber \notag\\
&+{\rm MSE}PP_sc\alpha(c+1)), \notag\\
d=&\sigma_v^4\alpha(\alpha c+c-1), \notag
\end{align}
and should satisfy:
$$
{\rm MSE}\beta P-\sigma_v^2c\alpha +\beta P-c\alpha P_s+(1-\beta )P_s\geq 0.
$$
Similarly, for $\beta$ fixed, the expression in \eqref{eq:avg_crb} can be used to determine the power of training symbol $P$ necessary to reach an attainable level of CRB equal to ${\rm MSE}$. Particularly, the $P$ that allows to achieve this level solves the following equation:
\begin{align}
    a_1P^2+b_1P+d_1=0,
\end{align}
where 
\begin{align}
a_1=&\frac{{\rm MSE}\beta^2}{cP_s}({\rm MSE}cP_s+\sigma_v^2(1-c)), \notag\\
b_1=&-\frac{\sigma_v^2\beta}{cP_s}((1-c)(\alpha\sigma_v^2-{\rm MSE}P_s+P_s\beta {\rm MSE})\nonumber \notag\\
&+c\alpha P_s{\rm MSE}(1+c)), \notag\\
d_1=&\sigma_v^4((1-c)(\beta\alpha-\alpha)+\alpha^2 c),\notag
\end{align}
and should also satisfy
$$
{\rm MSE}\beta P-\sigma_v^2c\alpha +\beta P-c\alpha P_s+(1-\beta )P_s\geq 0.
$$
To get more insights into the behavior of the average $\overline{\rm CRB}^{\rm avg}$, we will study its asymptotic value in two different limiting cases. The first case corresponds to the situation in which $\beta$ approaches $1$. It models the scenario where we use almost all symbols for training. The second case models the situation in which $\beta$ approaches  zero and $\frac{K}{N}=c\alpha=\theta \beta$ where $\theta\in(0,1)$.

\noindent{\underline{Case 1: $\beta\to 1$.} } It follows from \eqref{eq:m_tilde} that:
\begin{equation}
{\overline{m}}_{\tilde{\bf X}}(0)=\frac{1}{P}+\overline{f}(\beta), \label{eq:m_bar}
\end{equation}
where $\overline{f}=O(1-\beta)$. To further refine the calculations, we need to find an approximation of $\overline{f}(\beta)$ up to the order $(1-\beta)^2$. For that, plugging \eqref{eq:m_bar} into \eqref{eq:m_tilde}, we obtain:
$$
\beta P \overline{f}(\beta)+\frac{(1-\beta)P_s(\frac{1}{P}+\overline{f}(\beta))}{1+c\alpha P_s (\frac{1}{P}+\overline{f}(\beta))}=1-\beta,
$$
from which we obtain:
$$
\overline{f}(\beta)=\frac{1-\beta}{P}-\frac{(1-\beta)P_s}{P(P+c\alpha P_s)}+O((1-\beta)^2),
$$
and thus:
\begin{align}
\overline{\rm CRB}^{\rm avg}=&\frac{\sigma_v^2\alpha}{P}+(1-\beta)\frac{\sigma_v^2\alpha}{P}-\frac{(1-\beta)(1-c)\alpha\sigma_v^2P_s}{P(P+c\alpha P_s)}\nonumber \\
&+O((1-\beta)^2) \notag\\
=&\frac{\sigma_v^2\alpha}{P}+(1-\beta)\frac{\alpha\sigma_v^2}{P}(1-\frac{(1-c)P_s}{P+c\alpha P_s})+O(1-\beta)^2.\label{eq:func_P}
\end{align}
It follows from \eqref{eq:func_P} that the CRB is not always a decreasing function of $\beta$. Indeed, one can easily see that if $P\leq P_s(1-c-c\alpha)$, increasing $\beta$ and thus the number of training symbols will not reduce the CRB but would result in an increase of it. The reason lies in that when the power allocated to training symbols is relatively small compared to that allocated to unknown data symbols, the unknown data symbols carry more useful information for channel estimation than the low-power training symbols, making them more informative in this context. 

\noindent{\underline{Case 2: $\beta\to 0$ and $\frac{K}{N}=c\alpha=\theta\beta$}.} It follows from \eqref{eq:m_tilde} that in this case, 
\begin{equation}
\overline{m}_{\tilde{\bf X}}(0)=\frac{1}{P_s}+\overline{g}(\beta), \label{eq:m}
\end{equation}
where $\overline{g}(\beta)=O(\beta). $ To further refine the calculations, we need to find an approximation of $\overline{g}(\beta)$ up to the order $O(\beta^2)$. For that plugging \eqref{eq:m} into \eqref{eq:m_tilde}, we obtain:
$$
(1+P_s\overline{g}(\beta))(1-\theta\beta)-\beta+\frac{\beta P}{P_s}=1+O(\beta^2),
$$
from which we obtain:
$$
\overline{g}(\beta)=\frac{\beta(\theta+1)}{P_s}-\frac{\beta P}{P_s^2}+O(\beta^2),
$$
and thus
\begin{align}
\overline{\rm CRB}^{\rm avg}&=\sigma_v^2\frac{\theta}{P}+\alpha \sigma_v^2\beta\left(\frac{\theta+1}{P_s}-\frac{P}{P_s^2}\right)+O(\beta^2).
\end{align}

Similar to the previous case, we note that when $P\leq P_s(\theta+1)$, the asymptotic average CRB increases as more training symbols are used. This is because these are not of sufficient power, and it is better here to rely on unknown symbols with more power than low-power training symbols. 

Next, we move on studying the average CRB under the following asymptotic regime:
\begin{assumption}~\label{assumption_DCRB_2}
We assume that the number of users $K$ is fixed, while the transmission block length $N$ grows to infinity, i.e., $N \rightarrow \infty$. The number of transmit antennas $M$ also grows to infinity such that the ratio $\frac{M}{N} \rightarrow \alpha$ for some constant $\alpha > 0$. Under this asymptotic regime, we consider two cases:
\begin{itemize}
    \item \textbf{Case 1}: The number of training symbols $L$ is fixed.
    \item \textbf{Case 2}: The number of training symbols $L$ grows to infinity such that $\frac{L}{N} \rightarrow \beta\in (0,1)$.
\end{itemize}
\end{assumption}
\begin{theorem}\label{assumption 2} \label{th:crb_simple}
    Under Assumption~\ref{assumption_DCRB_2}, the average CRB satisfies the following convergence:
    \begin{align}
    {\rm CRB}^{\rm avg}-\overline{\rm CRB}^{\rm avg}\xrightarrow[]{a.s.} 0, \notag
    \end{align}
    while $\overline{\rm CRB}^{\rm avg}$ takes the following expression in Case 1
    \begin{equation}
    \overline{\rm CRB}^{\rm avg}=\frac{\sigma_v^2K}{LP}+\frac{\alpha \sigma_v^2}{P_s}, \label{eq:case1}
    \end{equation}
    and the following expression 
    \begin{equation}
    \overline{\rm CRB}^{\rm avg}=\frac{\alpha \sigma_v^2}{P_s(1-\beta)+\beta P},
    \label{eq:case2}
    \end{equation}
    in Case 2.
    \begin{proof}
The result follows from applying the strong law of large numbers to show that under Assumption \ref{assumption_DCRB_2}, we have:
    \begin{align}
    \left\|\frac{1}{N}{\bf X}-P_s\frac{(N-L)}{N}\mathbf{I}_K-P\frac{L}{N}\mathbf{I}_K\right\|\xrightarrow[]{a.s.}0.
    \end{align}
    By replacing $\frac{1}{N}{\bf X}$ by $P_s\frac{(N-L)}{N}\mathbf{I}_K+P\frac{L}{N}\mathbf{I}_K$ in the CRB, we get thus the desired. 
    \end{proof}
\end{theorem}

\noindent{\bf Comparison with the CRB for training-based schemes. }
Replacing $\tilde{\bf X}$ by $P\mathbf{I}_K$, the average CRB in the training-based scheme simplifies to:
\begin{align}
{\rm CRB}_{\rm training}^{\rm avg}=\frac{M\sigma_v^2}{LP}.
\end{align}
Considering the asymptotic regime in Case 1, we can note that ${\rm CRB}_{\rm training}^{\rm avg}$ grows to infinity linearly with $M$ as $M$ goes to infinity when $L$ is fixed. This must be compared with the $\overline{\rm CRB}^{\rm avg}$ which remains constant as $M$ and $N$ growing large to infinity with the same pace. The interest of semi-blind channel estimation appears clearly when $M$ is large as envisioned in massive MIMO systems. Furthermore, if the number of training symbols $L$ is supposed to be as large as $M$ as stated by Case 2, the average CRB in training based-schemes writes as:
\begin{align}
{\rm CRB}_{\rm training}^{\rm avg}=\frac{\alpha\sigma_v^2}{\beta P}.
\end{align}
By comparing the above expression with \eqref{eq:case2}, we can see that in semi-blind channel estimation adding more training symbols which translates into increasing $\beta$ is only beneficial if $P\geq P_s$. However, if $P_s$ is greater than $P$, it is better that we use only a small fraction of the available symbols for training and refining the channel estimation by exploiting the received signals from data transmissions. 

\noindent{\bf About the impact of the length $N$ of the whole transmission block.}
In~\cite{nayebi2017semi}, the authors concluded that as the number of antennas increases, the CRB for semi-blind estimation approaches the case where all $N$ data symbols are fully known. However, their conclusion was based on an inaccurate treatment of the Fisher information matrix in the asymptotic regime, where only element-wise convergence of each entry was considered. From Theorem \ref{th:crb_simple}, it follows that this result holds only when $L$ also grows with $N$. In this case, assuming $P = P_s$, $\overline{\rm CRB}^{\rm avg}$ derived under Case 2 simplifies to:
\begin{align}
\overline{\rm CRB}^{\rm avg} = \frac{\alpha \sigma_v^2}{P},
\end{align}
which effectively coincides with the scenario where all $N$ symbols are used for training. However, if $L$ is fixed and on the same order as the number of users, $\overline{\rm CRB}^{\rm avg}$ will always remain greater than $\frac{\sigma_v^2 K}{L P}$ and will not decrease indefinitely as $N$ increases.

\section{Analysis of stochastic Cram\'{e}r Rao Bound for semi-blind channel estimation}\label{sec_SCRB}
\subsection{Derivation of the stochastic Cram\'{e}r Rao bound} 
The derivation of the stochastic CRB is carried out assuming that the unknown data symbols are drawn from a probability distribution. Following \cite{nayebi2017semi}, we consider the case where the data symbols are modeled as independent random variables with mean zero and variance ${P_s}$. Under this scenario, the received symbols are modeled during training and data transmissions as:
\begin{equation}
\begin{split}
    &\mathbf{y}(n)\sim\mathcal{CN}(\mathbf{Gs}(n),\sigma_v^2\mathbf{I}_M),\,\,n = 0,\ldots,L-1, \\
    &\mathbf{y}(n)\sim\mathcal{CN}(\pmb{0},\mathbf{R}),\qquad\qquad n = L,\ldots,N-1,
\end{split}
\label{eq:data_model}
\end{equation}
where $\mathbf{R}=P_s\mathbf{GG}^H+\sigma_v^2\mathbf{I}_M$. The following theorem provides expressions for the stochastic CRB. While a similar result was established in \cite{nayebi2017semi}, our formulation is simpler due to the use of the Kronecker product. 
\begin{theorem}[\cite{nayebi2017semi}]
\label{th:stochastic_crb_derivation}
Assuming the received symbols are modeled as \eqref{eq:data_model}, the stochastic CRB for channel estimation is given by:
\begin{align}
    \mathbf{CRB}_s =& \left(\frac{2PL}{\sigma_v^2}\mathbf{I}_{2KM}+2P_s^2(N-L)(\mathcal{R}_1(\mathbf{T})+\mathcal{R}_2(\mathbf{C}))\right)^{-1},
\end{align}
where 
\begin{equation}\label{R1(T),R2(C)}
\begin{split}
    \mathcal{R}_1(\mathbf{T}) :=& \left[\begin{array}{cc}
        \mathrm{Re}\{\mathbf{T}\} & \mathrm{Im}\{\mathbf{T}\} \\
        -\mathrm{Im}\{\mathbf{T}\} & \mathrm{Re}\{\mathbf{T}\}
    \end{array}\right], \\
    \mathcal{R}_2(\mathbf{C}) :=& \left[\begin{array}{cc}
        \mathrm{Re}\{\mathbf{C}\} & -\mathrm{Im}\{\mathbf{C}\} \\
        -\mathrm{Im}\{\mathbf{C}\} & -\mathrm{Re}\{\mathbf{C}\}
    \end{array}\right],
\end{split}
\end{equation}
with $\mathbf{T}=(\mathbf{G}^H\mathbf{R}^{-1}\mathbf{G})\otimes(\mathbf{R}^T)^{-1}$, and 
\begin{align}
    \mathbf{C} = \left[\begin{array}{ccc}
        \mathbf{C}_{11} & \cdots & \mathbf{C}_{1K} \\
        \vdots & \ddots & \vdots \\
        \mathbf{C}_{K1} & \cdots & \mathbf{C}_{KK} \\
    \end{array}\right].
\end{align}
Here, $\mathbf{C}_{ij} = (\mathbf{g}_j^H\mathbf{R}^{-1})^T\mathbf{g}_i^H\mathbf{R}^{-1}$, the symbols $\mathcal{R}_1(\cdot)$ and $\mathcal{R}_2(\cdot)$ refer to the two block structures defined in \eqref{R1(T),R2(C)}. 
\end{theorem}
\begin{proof}
See Appendix~\ref{app:stochastic_crb_derivation}. 
\end{proof}
\noindent{\bf Stochastic average CRB. } Similar to the deterministic CRB, we define the stochastic average CRB as:
$$
{\rm CRB}_s^{\rm avg}:=\frac{1}{K}{\rm tr}({\bf CRB}_s).
$$
The stochastic average CRB has a similar form to that in \cite{nayebi2017semi}, but it remains mathematically complex. At first glance, it appears mathematically intractable for asymptotic performance analysis due to its implicit dependence on the singular vectors and singular values of the channel matrix ${\bf G}$. Contrary to this initial impression, we demonstrate in the following that the stochastic CRB actually depends exclusively on the singular values of the matrix ${\bf G}$, which enables further simplification and facilitates asymptotic analysis. Letting $\mathbf{G}=\mathbf{U}_{G}\mathbf{D}_{G}\mathbf{V}_{G}^H$, a singular value decomposition of ${\bf G}$ where $\mathbf{U}_{G}\in\mathbb{C}^{M\times M}$ and $\mathbf{V}_{G}\in\mathbb{C}^{K\times K}$ are the left and right singular vector matrices, and $\mathbf{D}_{G}=[\pmb{\Sigma}^T,\mathbf{0}^T]^T$ with $\pmb{\Sigma}=\mathrm{diag}([\sigma_1,\ldots,\sigma_K])\in\mathbb{C}^{K\times K}$ is a diagonal matrix whose diagonal elements are the singular values of matrix $\mathbf{G}$. By expressing ${\bf R}$ as ${\bf R}={\bf U}_G{\bf R}_G{\bf U}_G^{H}$ where ${\bf R}_G=P_s\mathbf{D}_G\mathbf{D}_G^H+\sigma_v^2\mathbf{I}_M$, we obtain the following proposition:


\begin{proposition}\label{tr(CRB)=tr(A1)+tr(A2)}
    ${\rm CRB}_s^{\rm avg}=\frac{1}{K}\mathrm{tr}(\mathbf{A}_1^{-1})+\frac{1}{K}\mathrm{tr}(\mathbf{A}_2^{-1})$, where 
    \begin{align}
    \mathbf{A}_1 =& \frac{2PL}{\sigma_v^2}\mathbf{I}_{KM}+2P_s^2(N-L)(\mathbf{T}_G+\mathbf{C_G}), \\
    \mathbf{A}_2 =& \frac{2PL}{\sigma_v^2}\mathbf{I}_{KM}+2P_s^2(N-L)(\mathbf{T}_G-\mathbf{C_G}),
    \end{align}
    with 
    \begin{align}
    \mathbf{T}_{G} =&(\mathbf{D}_{G}^H\mathbf{R}_{G}^{-1}\mathbf{D}_{G})\otimes(\mathbf{R}_{\mathbf{G}})^{-1}, \label{eq:TG}\\
    \mathbf{C}_G =& \left[\begin{array}{ccc}
    \mathbf{C}_{G_{11}} & \cdots & \mathbf{C}_{G_{1K}} \\
    \vdots & \ddots & \vdots \\
    \mathbf{C}_{G_{K1}} & \cdots & \mathbf{C}_{G_{KK}} \\
    \end{array}\right], \label{eq:CG}
    \end{align}
    with $\mathbf{C}_{G_{ij}}=\sigma_j\sigma_i{\bf R}_G^{-1}{\bf e}_j{\bf e}_i^{T}{\bf R}_G^{-1}$ being a all-zero matrix except $c_{ij}=\frac{\sigma_j\sigma_i}{(P_s\sigma_j^2+\sigma_v^2)(P_s\sigma_i^2+\sigma_v^2)}$ at the $j$th row and $i$th column, and ${\bf e}_1,\cdots,{\bf e}_K$  the canonical basis vectors in $\mathbb{R}^{M}$.
    \label{prop:stochastic_Crb}
    \end{proposition}
    \begin{proof}
    See Appendix~\ref{app:simplified_crb}.
    \end{proof}
As can be seen from Proposition \ref{prop:stochastic_Crb}, the stochastic CRB depends only on the singular values of matrix ${\bf G}$. It involves computing the inverse of matrices ${\bf A}_1$ and ${\bf A}_2$. Both of these matrices are not diagonal due to the presence of matrix ${\bf C}_G$. However, since ${\bf C}_G$ possesses only a few non-zero elements, closed-form expressions for the trace of inverse of matrices ${\bf A}_1$ and ${\bf A}_2$ can be computed in closed-form, to obtain finally the following expression of the average stochastic CRB. 
\begin{theorem}
For $x>0$, define the empirical Stieltjes transform of the matrix ${\bf G}^{H}{\bf G}$ as:
$$
m_{{\bf G}^{H}{\bf G}}(x)=\frac{1}{K}\sum_{i=1}^{K}\frac{1}{\sigma_i^2-x},
$$
where $\sigma_1^2\leq \sigma_2^2\leq \cdots\leq \sigma_M^2$ being the eigenvalues of ${\bf G}^{H}{\bf G}$. 
Then, the average stochastic CRB is written as in \eqref{eq:stochastic_crb}.
\label{th:stochastic_crb}
\end{theorem}
\begin{figure*}[t]
\begin{align}
{\rm CRB}_s^{\rm avg}=&\frac{\sigma_v^2K}{PL}+\frac{(M-K)\sigma_v^2}{L(P-P_s)+P_sN}+\frac{(M-K)(N-L)\sigma_v^4}{(L(P-P_s)+P_sN)^2}m_{{\bf G}^{H}{\bf G}}(-\frac{\sigma_v^2PL}{P_s(NP_s+L(P-P_s))})\notag\\
&-\sum_{j=1}^K \frac{\sigma_v^4P_s(N-L)\sigma_i^2}{PL(PLP_s\sigma_i^2+L\sigma_v^2(P-P_s)+P_sN\sigma_v^2)}m_{{\bf G}^{H}{\bf G}}(-\frac{\sigma_v^2(PL\sigma_v^2+P_s^2N\sigma_i^2+P_s\sigma_i^2L(P-P_s))}{PLP_s^2\sigma_i^2+P_s\sigma_v^2(NP_s+(P-P_s)L)}), \label{eq:stochastic_crb} \\
{\rm CRB}_{\rm gain}^{\rm avg}=&\frac{\sigma_v^2(M-K)(N-L)}{(L(P-P_s)+P_sN)}\Big(\frac{P_s}{LP}-\frac{\sigma_v^2}{L(P-P_s)+P_sN}m_{{\bf G}^{H}{\bf G}}(-\frac{\sigma_v^2PL}{P_s(NP_s+L(P-P_s))})\Big)\nonumber\\
&+\sum_{j=1}^K \frac{\sigma_v^4P_s(N-L)\sigma_i^2}{PL(PLP_s\sigma_i^2+L\sigma_v^2(P-P_s)+P_sN\sigma_v^2)}m_{{\bf G}^{H}{\bf G}}(-\frac{\sigma_v^2(PL\sigma_v^2+P_s^2N\sigma_i^2+P_s\sigma_i^2L(P-P_s))}{PLP_s^2\sigma_i^2+P_s\sigma_v^2(NP_s+(P-P_s)L)})\label{eq:crb_gain}.
\end{align}
\hrule
\end{figure*}

\begin{proof}
See Appendix~\ref{app:stochastic_crb}. 
\end{proof}
\noindent{\bf Comparison with training-based schemes.} Recall that the average CRB of training-based systems is given by:
\begin{align}
{\rm CRB}_{\rm training}^{\rm avg}=\frac{\sigma_v^2M}{LP}.
\end{align}
Hence, we can decompose the stochastic average CRB as:
\begin{align}
{\rm CRB}_{s}^{\rm avg}={\rm CRB}_{\rm training}^{\rm avg}-{\rm CRB}_{\rm gain}^{\rm avg},
\end{align}
where ${\rm CRB}_{\rm gain}^{\rm avg}$ represents the gain obtained from using semi-blind channel estimation and is given by \eqref{eq:crb_gain}. Using the fact that $m_{{\bf G}^{H}{\bf G}}(-x)\leq \frac{1}{x}$ for $x\geq 0$, we conclude that ${\rm CRB}_{\rm gain}^{\rm avg}$ is positive and hence the average stochastic CRB for semi-blind channel estimation is always lower than the average CRB of training-based schemes. This fact is consistent with our findings related to deterministic CRB. 
\subsection{Asymptotic analysis of the stochastic CRB}
In this section, we analyze the average stochastic CRB by establishing a direct connection to the Stieltjes transform of $\mathbf{G}^H\mathbf{G}$ under different asymptotic growth regimes. The first regime we examine corresponds to the one outlined in Assumption \ref{assumption_dcrb_1}. However, unlike the deterministic CRB, the stochastic CRB is influenced by the singular values of the channel matrix. As a result, its behavior is directly related to the statistical properties of the channel. To carry out our performance analysis, we adopt the following channel distribution model:
\begin{assumption}[Channel model]
Matrix ${\bf G}$ is composed of independent and identically distributed entries with mean zero and variance $\frac{1}{M}$ and finite fourth order moment. \label{ass:channel_model}
\end{assumption}

It is well-known from~\cite{marchenko1967distribution} that as $M\rightarrow+\infty$, $K\rightarrow+\infty$ with $K/M\rightarrow c\in(0,1)$, under Assumption \ref{ass:channel_model}, the empirical Stieltjes transform of matrix ${\bf G}^{H}{\bf G}$ converges to the Stieltjes transform of the marchenko Pastur law which is given by:
\begin{align}
    m_{F}(z) = \frac{1-c-z-\sqrt{(1-c-z)^2-4cz}}{2cz}. 
\end{align}
in the sense that:
\begin{align}\label{as convergence}
    m_{{\mathbf{G}}^H{\mathbf{G}}}(z) \xrightarrow[M\rightarrow\infty]{\mathrm{a.s.}} m_{F}(z),
\end{align}
Moreover, the Marchenko-Pastur law denoted by $\mu(x)$ is given by:
\begin{align}
\mu(x) = \frac{\sqrt{(x-a)(b-x)}}{2c\pi x}\mathbf{1}_{[a,b]}(x),
\end{align}
where $\mathbf{1}_{[a,b]}(x)$ is the indicator function of the interval $[a,b]$, $a=(1-\sqrt{c})^2$ and $b=(1+\sqrt{c})^2$.  
Based on this result, we can derive the asymptotic expression for the average stochastic CRB as follows. 
\begin{theorem}\label{sto_assumption1}
    Under Assumption~\ref{assumption_dcrb_1} and Assumption \ref{ass:channel_model}, the stochastic average CRB satisfies:
    $$
    {\rm CRB}_{s}^{\rm avg}-\overline{\rm CRB}_s^{\rm avg}\xrightarrow[]{a.s.} 0,
    $$
    where 
    \begin{align}
    &\overline{\rm CRB}_s^{\rm avg}= \frac{\sigma_v^2c\alpha}{P\beta} +\frac{\sigma_v^2(1-c)\alpha}{(P-P_s)\beta+P_s} \notag \\
    &\qquad+\frac{\sigma_v^4(1-\beta)(1-c)\alpha}{((P-P_s)\beta+P_s)^2}m_F\left(-\frac{\sigma_v^2P\beta}{P_s((P-P_s)\beta+P_s)}\right) \notag \\
    &\qquad-\int_a^b\frac{c\alpha(1-\beta)P_s x\sigma_v^4}{P\beta(P\beta P_s x+((P-P_s)\beta+P_s)\sigma_v^2)} \notag\\
   &\qquad\cdot m_F\left(-\frac{\sigma_v^2(P_s((P-P_s)\beta+P_s)x+P\beta\sigma_v^2)}{P\beta P_s^2x+P_s((P-P_s)\beta+P_s)\sigma_v^2}\right)\mu(x)dx .
    \end{align}
    \begin{proof}
        It is well-established from standard results in random matrix theory that \( m_{{\bf G}^{H}{\bf G}}(z) \) converges uniformly to \( m_F(z) \) over any compact interval of \( \mathbb{C} \setminus \mathbb{R}_{+} \) \cite{couillet2011random}. Furthermore, for any bounded and continuous function \( f \),
\begin{equation}
\frac{1}{K}\sum_{i=1}^{K} f(\lambda_i) \xrightarrow[]{a.s.} \int_{a}^{b} f(x)\mu(dx),
\end{equation}
almost surely. The proof proceeds by leveraging both of these properties in the analysis of the CRB expression in \eqref{eq:stochastic_crb}.

    \end{proof}
\end{theorem}

Similar to the deterministic CRB, to gain insights into the behavior of the average stochastic CRB, we consider studying its behavior corresponds to the case where $\beta$ approaches $1$, modeling the situation of almost use of training symbols for channel estimation. By Taylor-expanding the limiting average stochastic CRB around $\beta$ approaching one, we can obtain after easy calculations that:
\begin{align}
\overline{\rm CRB}_s^{\rm avg}=&\frac{\sigma_v^2\alpha}{P}+\frac{\sigma_v^2\alpha}{P}(1-\beta)(1-\frac{P_s}{P}) + \frac{c\alpha\sigma_v^2}{P^2}P_s(1-\beta)\notag\\
&-\frac{(1-\beta)}{P^2}\sigma_v^4c\alpha m_F(-\frac{\sigma_v^2}{P_s})(1-\frac{\sigma_v^2}{P_s}m_F(-\frac{\sigma_v^2}{P_s})) \notag\\
&+O(1-\beta)^2.
\end{align}
Using the fact that $xm_F(-x)<1$, and that $m_F(-x)\leq \frac{1}{x}$, we can easily see that when $P\geq P_s$, increasing the number of training symbols always results in a decrease in the average stochastic CRB. 

Next, we go on studying the average stochastic CRB under the regime specified in Assumption \ref{assumption_DCRB_2} and the following channel model. 
\begin{assumption}(Channel model.) Matrix ${\bf G}$ is modeled as ${\bf G}={\bf H}{\bf B}^{1/2}$ where ${\bf B}$ models the large-scale fading coefficients between the BS and the users and ${\bf H}$ is composed of independent and identically distributed entries with mean zero and variance $\frac{1}{M}$ and finite fourth order moment. Moreover, the spectral norm of matrix ${\bf B}$ is assumed to be uniformly bounded below and above in $M$, that is $\limsup_{M} \lambda_M^2<\infty$ and $\liminf_{M}\lambda_1^2>0$ with $\lambda_1^2\leq\cdots\leq \lambda_M^2$ being the eigenvalues of ${\bf B}$.  \label{ass:channel_model_cl}
\end{assumption}

\begin{theorem}
Under Assumption \ref{assumption_DCRB_2} and Assumption \ref{ass:channel_model_cl}, the stochastic average CRB satisfies the following convergence:
$$
{\rm CRB}_s^{\rm avg}-\overline{\rm CRB}_s^{\rm avg}\xrightarrow[]{a.s.}0,
$$
where 
\begin{equation}
\overline{\rm CRB}_s^{\rm avg}=\frac{\sigma_v^2K}{PL}+\frac{\sigma_v^2\alpha}{P_s}+\frac{\alpha \sigma_v^4}{P_s^2}m_{{\bf B}}(0)-\sum_{i=1}^K \frac{\lambda_i^2\sigma_v^2}{PL}m_{{\bf B}}(-\lambda_i^2),\label{eq:case1s}
\end{equation}
in Case 1 and
\begin{align}
\overline{\rm CRB}_s^{\rm avg}=&\frac{\alpha \sigma_v^2}{\beta(P-P_s)+P_s}+\frac{\alpha (1-\beta)\sigma_v^4}{(\beta(P-P_s)+P_s)^2}\nonumber\\
&\cdot m_{{\bf B}}(-\frac{\sigma_v^2P\beta}{P_s(P_s+\beta(P-P_s))}),\label{eq:case2s}
\end{align}
in Case 2, 
where $m_{\bf B}(x)$ is the Stieltjes transform of the empirical distribution of the eigenvalues of ${\bf B}$ given by:
$$
m_{\bf B}(x)=\frac{1}{K}\sum_{i=1}^{K}\frac{1}{\lambda_i^2-x},
$$
with $\lambda_1^2,\cdots,\lambda_M^2$ being the eigenvalues of matrix ${\bf B}$.
\begin{proof}
Applying the strong law of large numbers yields
    \begin{align}
    \|\mathbf{G}^H\mathbf{G}-\mathbf{B}\|\xrightarrow[]{a.s.}0.
    \end{align}
    By replacing $\mathbf{G}^H\mathbf{G}$ by $\mathbf{B}$ in \eqref{eq:stochastic_crb}, we get thus the desired.
\end{proof}
\end{theorem}
\begin{remark}
In practice, it is important for channel models to incorporate spatial correlation for more accurate and realistic performance analysis. To this end, we adopt the widely used Kronecker model~\cite{forenza2007simplified}, where the channel matrix is modeled as ${\bf G}={\bf \Phi}_{r}^{1/2}{\bf H}{\bf \Phi}_{t}^{1/2}$, where ${\bf \Phi}_{t}$ and ${\bf \Phi}_{r}$ denote the spatial correlation matrices at the transmitter and receiver, respectively, capturing the spatial dependencies among the array elements. Under this model, it can be shown that
\begin{align}
    \left\|\mathbf{G}^H\mathbf{G}-\pmb{\Phi}_t\frac{1}{K}\mathrm{tr}(\pmb{\Phi}_r)\right\|\xrightarrow[M\rightarrow\infty]{a.s.} 0.
\end{align}
Based on this result, the asymptotic CRB expressions for spatially correlated MIMO channels can be derived by following a procedure similar to that used in Theorem 8.
\end{remark}
\noindent{\bf Comparison with the deterministic CRB.} Under Assumption 3, we compare the asymptotic expressions for the average stochastic CRB in \eqref{eq:case1s} and \eqref{eq:case2s} with the deterministic CRB in \eqref{eq:case1} and \eqref{eq:case2}. We observe that the stochastic CRB is numerically larger, indicating a tighter lower bound on the MSE of any unbiased channel estimator within the same estimation framework.

\noindent{\bf Comparison with training-based schemes.} Consistent with our findings on the deterministic CRB, the stochastic CRB approaches zero as $N$ increases to infinity at a faster rate than $M$, but only if $L$ grows at the same pace (Case 2 in Assumption \ref{assumption_DCRB_2}). However, when the number of training symbols remains on the order of $K$, the average stochastic CRB is bounded below by $\frac{\sigma_v^2K}{PL}$, preventing it from decreasing indefinitely as $N$ increases.

\par 

\begin{figure}[t]
\centering
\includegraphics[width=2.7 in]{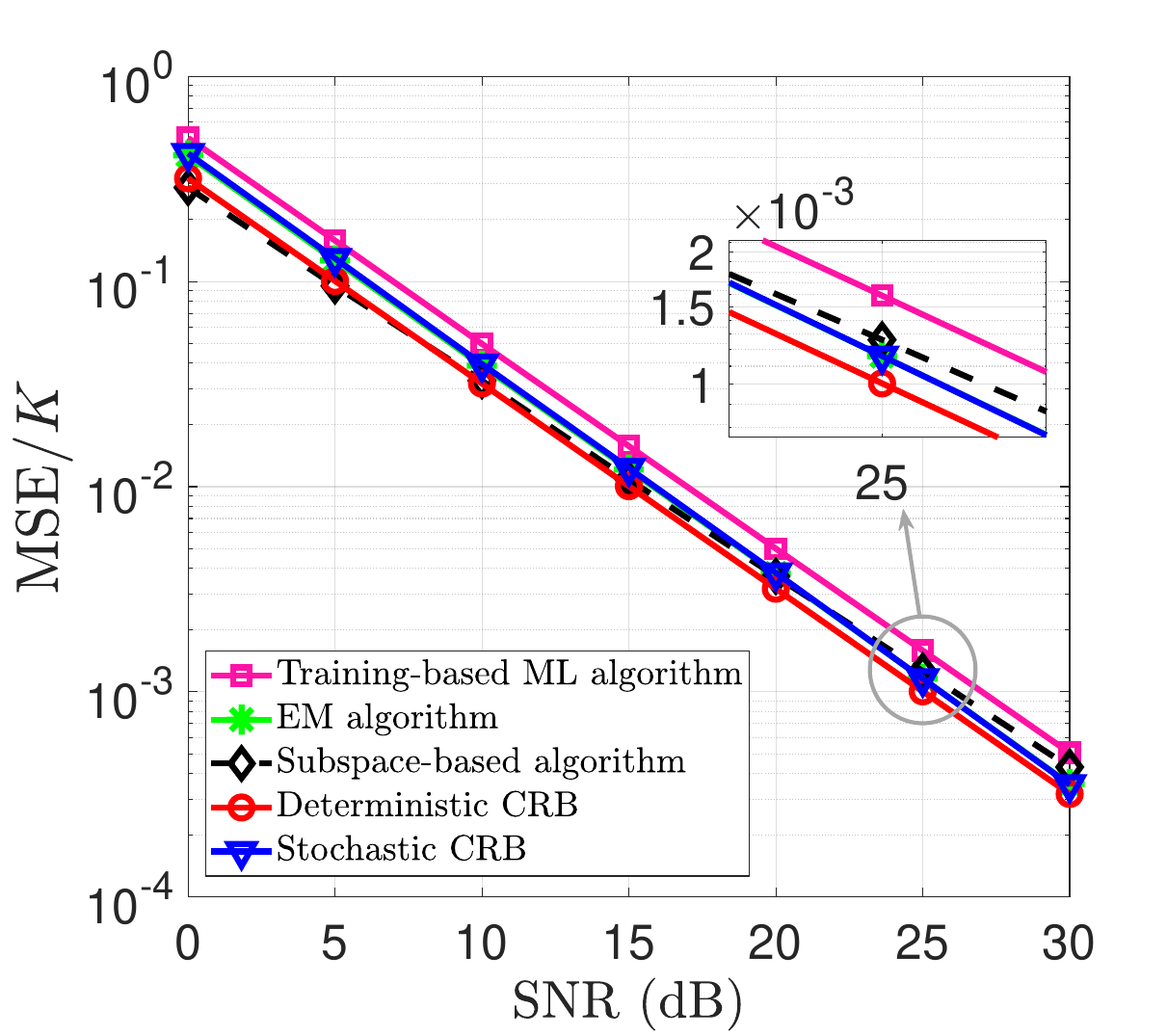}
\caption{$\mathrm{MSE}/K$ versus $\mathrm{SNR}$ with $K=32$, $L=64$, $M=512$, and $N=1024$.}
\label{SNR_MSE}
\end{figure}

\begin{figure}[t]
\centering
\subfloat[Deterministic CRB under Assumption~\ref{assumption_dcrb_1}]{
\label{DCRB_Accuracy}
\includegraphics[width=2.7 in]{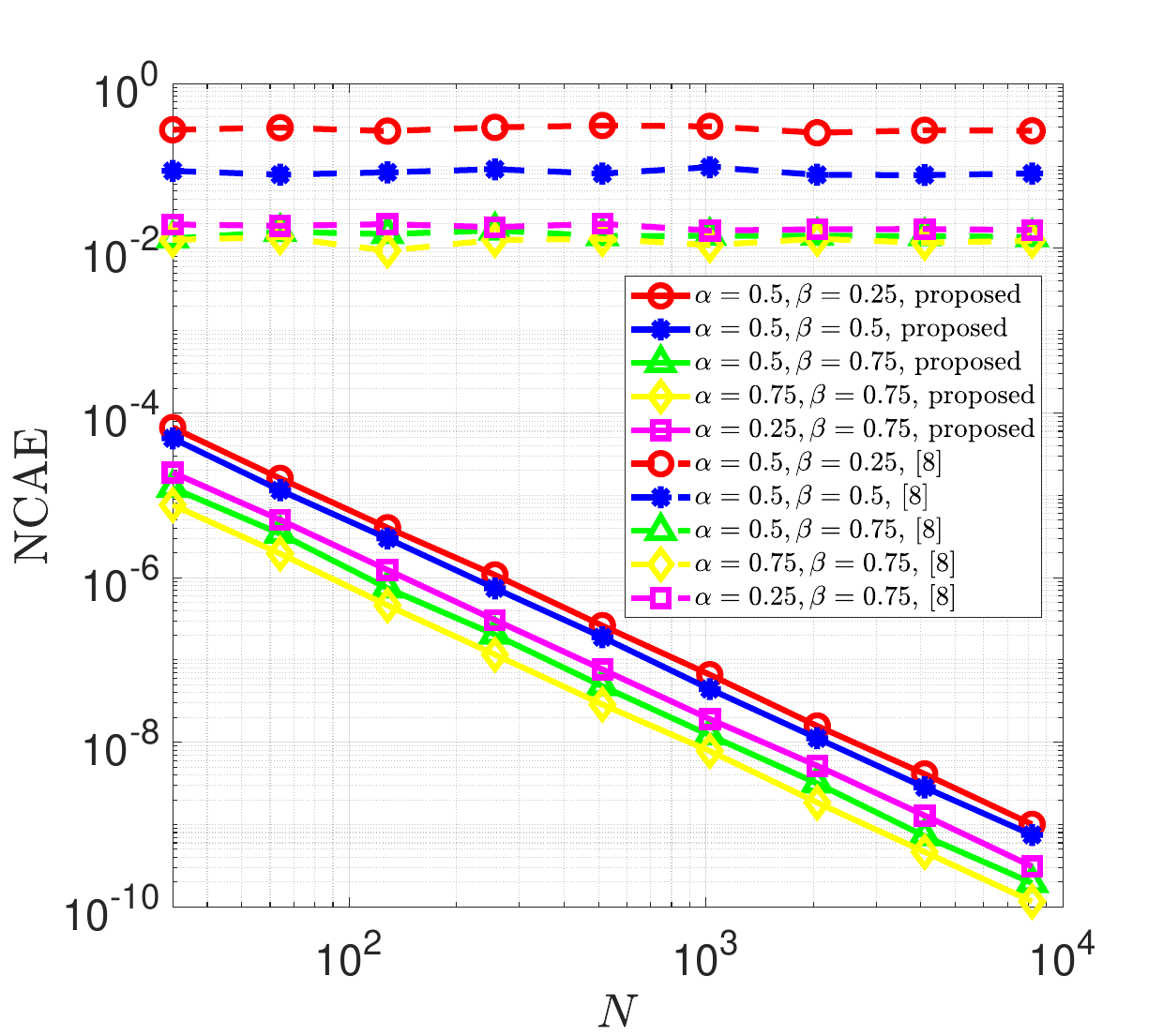}}
\\
\subfloat[Stochastic CRB under Assumption~\ref{assumption_dcrb_1}]{
\label{SCRB_Accuracy}
\includegraphics[width=2.7 in]{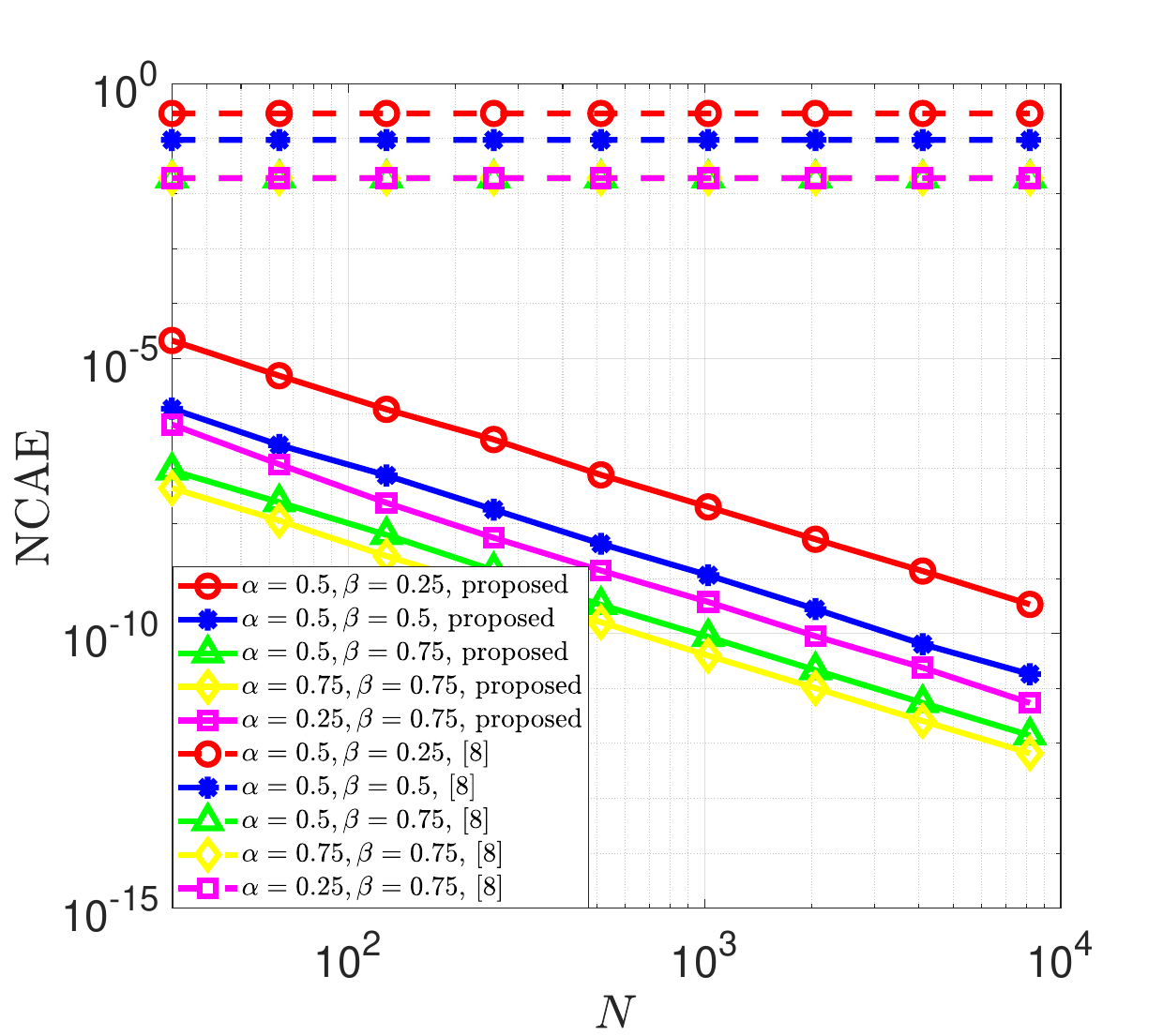}}
\caption{$\mathrm{NCAE}$ versus $N$ with $K=M/2$, $\mathrm{SNR}=10$\,dB.}
\label{Accuracy_CRB}
\end{figure}

\section{Numerical results} 
\label{sec:sim}
This section provides numerical simulations illustrating the asymptotic behaviors of deterministic CRB and stochastic CRB given in the previous sections. In the simulation experiments, two types of users are uniformly distributed within the single cell with radius $100$\,m and $500$\,m, respectively. For large-scale fading coefficients, we consider the normalized three-slope path loss model~\cite{tang2001mobile} as follows
\begin{align}
    \rho_k=\frac{\bar{\beta}_k}{\mathbb{E}\{\bar{\beta}_k\}},\,\,\, \bar{\beta}_k = \left\{\begin{array}{ll}
       c_0  & d_k \le d_0 \\
       \frac{c1}{d_k^2}  & d_0<d_k\le d_1 \\
       \frac{c_2z_k}{d_k^{3.5}} & d_k >d_1
    \end{array}\right. ,
\end{align}
where $d_k$ is the distance between the BS and $k$th user, $z_k$ is the log-normal shadow fading, i,e., $10\mathrm{log}_{10}z_k\sim\mathcal{N}(0,\sigma_{\mathrm{shad}}^2)$ with $\sigma_{\mathrm{shad}}=8$\,dB, and 
\begin{align}
    10\mathrm{log}_{10}c_2 =& -46.3-33.9\mathrm{log}_{10}f + 13.82\mathrm{log}_{10}h_B \notag \\
    &+ (1.1\mathrm{log}_{10}f-0.7)h_R - (1.56\mathrm{log}_{10}f-0.8), \notag\\
    10\mathrm{log}_{10}c_1 =& 10\mathrm{log}_{10}c2 - 15\mathrm{log}_{10}(d_1), \notag \\
    10\mathrm{log}_{10}c_0 =& 10\mathrm{log}_{10}c1 - 20\mathrm{log}_{10}(d_0),\notag
\end{align}
with $d_0 = 10$\,m, $d_1=50$\,m, $f=1900$\,MHz being the carrier frequency, $h_B=15$\,m being the BS antenna height, $h_R$ being the user antenna height. Additionally, the signal-to-noise ratio (SNR) is calculated by $\mathrm{SNR}=\sigma_v^{-2}$. We assume that matrix ${\bf S}_p$ satisfies ${\bf S}_p{\bf S}_p^{H}=\mathbf{I}_K$. Furthermore, for the deterministic CRB, we assume that ${\bf S}_d$ is composed of symbols drawn from a QPSK constellation with $P_s=1$. We assess the accuracy of the asymptotic approximation to the CRB using the normalized CRB approximation error (NCAE), defined as
\begin{align}
    \mathrm{NCAE} = \frac{\mathbb{E}\big[|{\rm CRB}_{\rm true}-\overline{\rm CRB}_{\rm asy}|^2\big]}{(\mathbb{E}\big[{\rm CRB}_{\rm true}^2\big])},
\end{align}
where ${\rm CRB}_{\mathrm{true}}$ denotes either the deterministic or stochastic CRB computed using (\ref{eq:deterministic_crb}) or (\ref{eq:stochastic_crb}), respectively, and ${\overline{\rm CRB}}_{\mathrm{asy}}$ represents the corresponding asymptotic expression.

\begin{figure}[t]
\centering
\includegraphics[width=2.7 in]{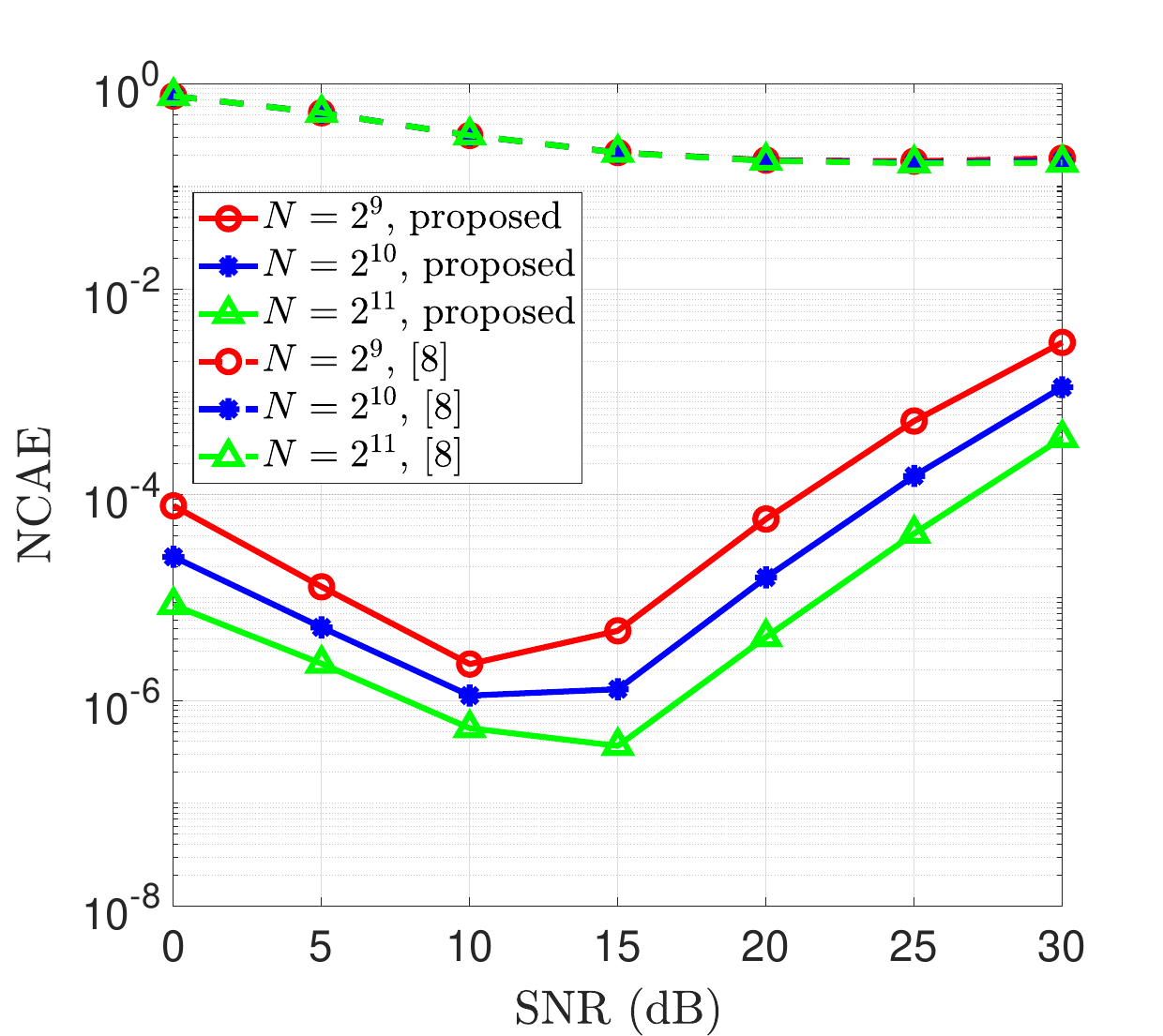}
\caption{$\mathrm{NCAE}$ versus $\mathrm{SNR}$ under Assumption~\ref{assumption_DCRB_2}-Case 1 with $\alpha=0.5$, $\beta=0.25$, and $K=8$.}
\label{SCRB_SNR}
\end{figure}

\begin{figure}[t]
\centering
\includegraphics[width=2.7 in]{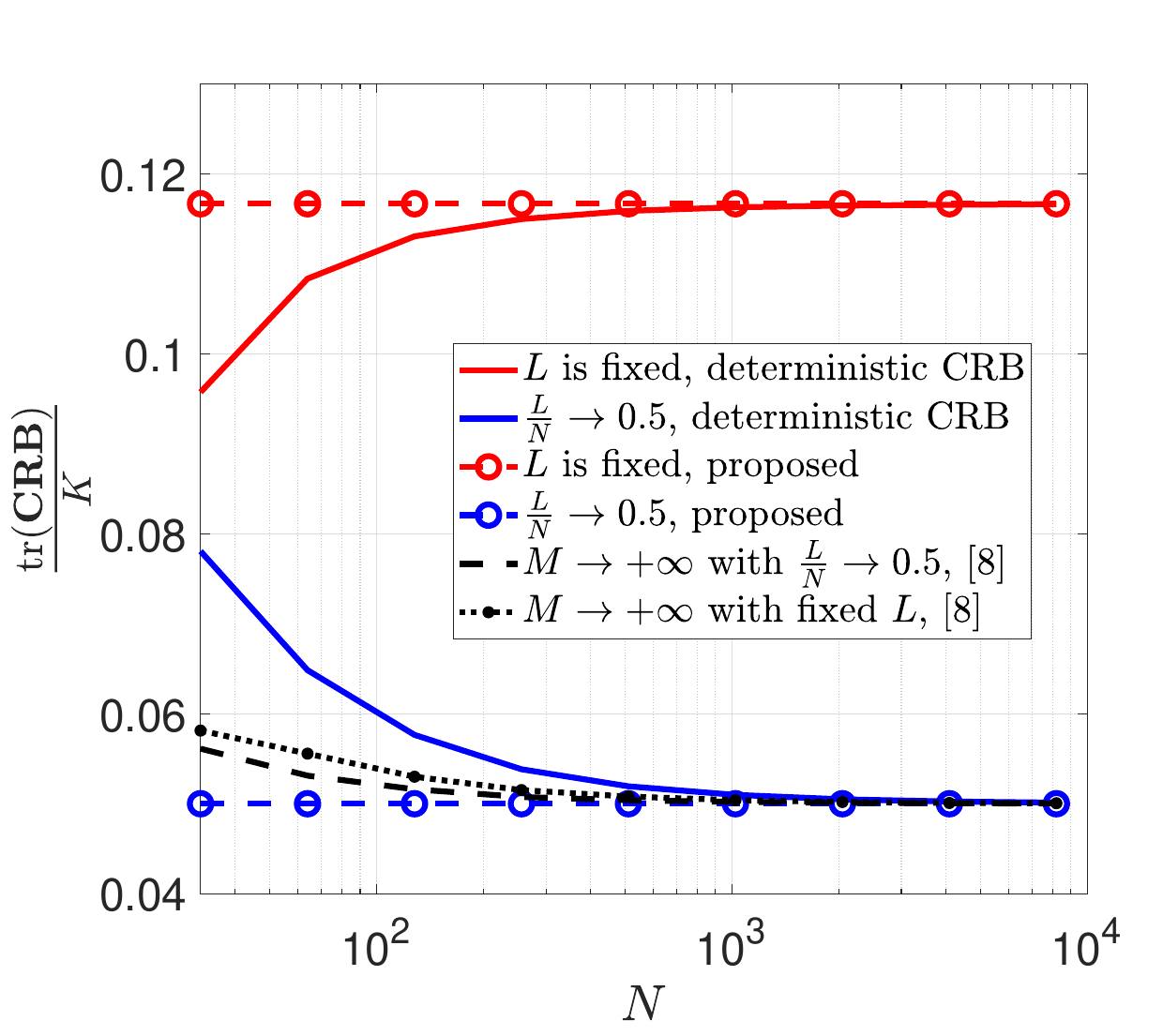}
\caption{Deterministic CRB versus $N$ under Assumption~\ref{assumption_DCRB_2} with $\mathrm{SNR}=10$\,dB, $\alpha=0.5$, and $K=8$.}
\label{DCRB_3}
\end{figure}

\begin{figure}[t]
\centering
\subfloat[Under Assumption~\ref{assumption_DCRB_2}-Case 1]{
\label{SCRB_3_Case_1}
\includegraphics[width=2.7 in]{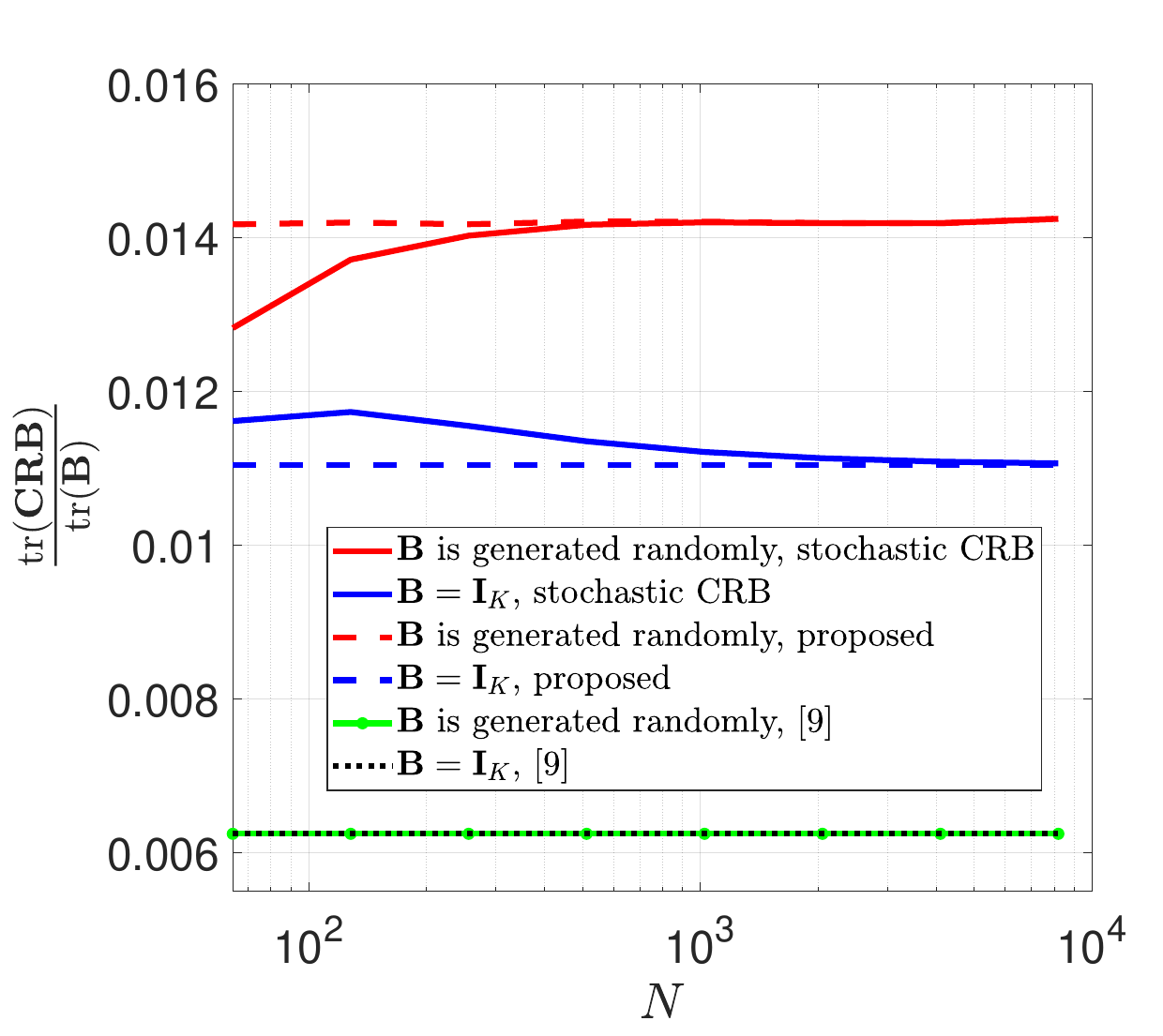}}
\\
\subfloat[Under Assumption~\ref{assumption_DCRB_2}-Case 2]{
\label{SCRB_3_Case_2}
\includegraphics[width=2.7 in]{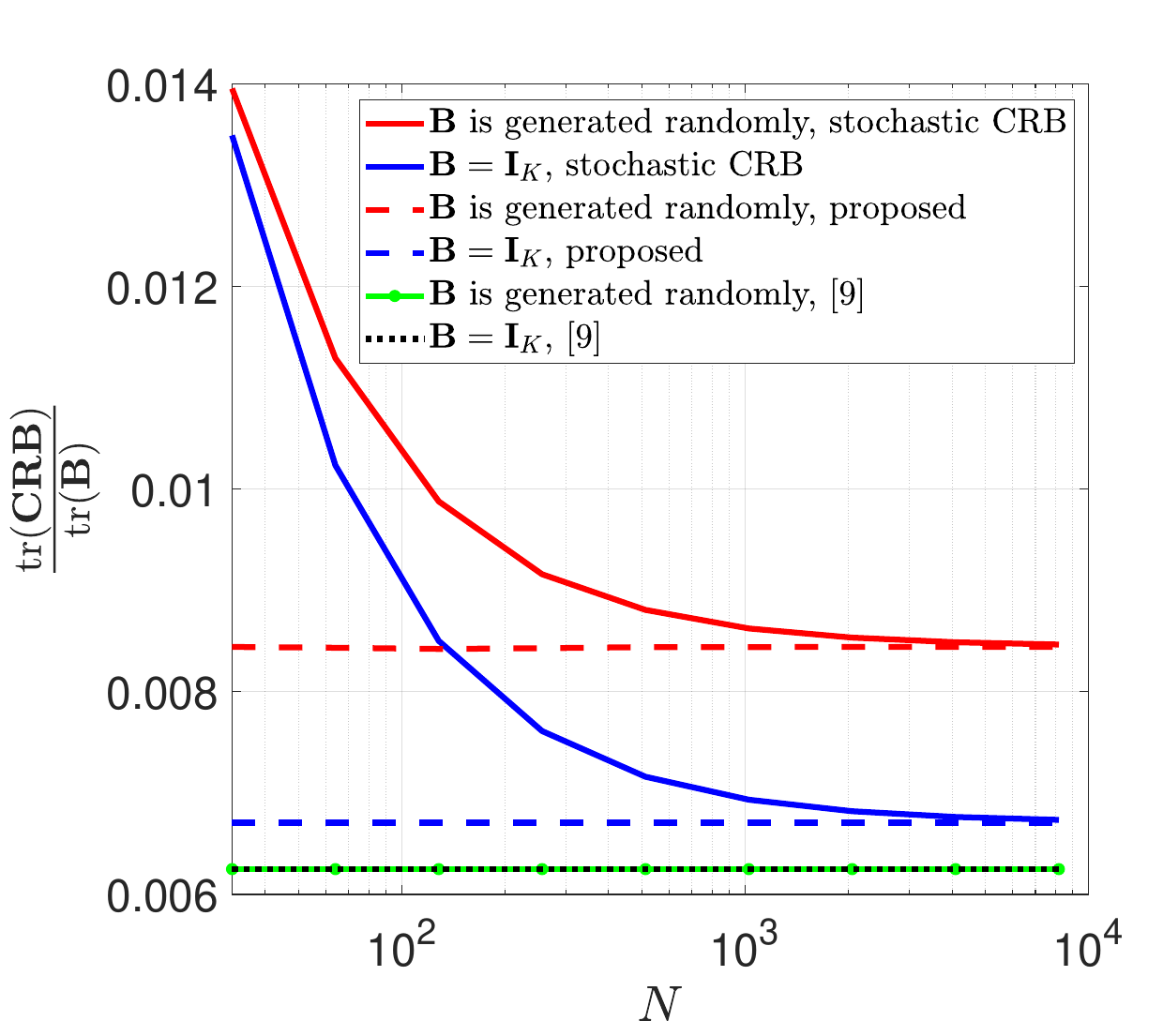}}
\caption{Stochastic CRB for different $\mathbf{B}$.}
\label{CRB}
\end{figure}

\subsection{Comparison Between the CRBs and Channel Estimation Errors of Practical Algorithms}
In the first experiment, we examine the impact of increasing SNR on the performance of selected estimators and the exact CRB. Specifically, we consider both the deterministic CRB (\ref{eq:deterministic_crb}) and the stochastic CRB (\ref{eq:stochastic_crb}), along with three estimators: the training-based ML estimator given by $\hat{G}_{\mathrm{training}}=\frac{1}{P}\mathbf{Y}_p\mathbf{S}_p^H$, the EM algorithm~\cite{nayebi2017semi}, and the subspace-based algorithm~\cite{zhong2024subspace}. We set the parameters to $N=1024$, $M=512$, $L=64$, and $K=32$, respectively. As shown in Fig.~\ref{SNR_MSE}, although the existing semi-blind channel estimation algorithms maybe biased, the CRBs still provide useful reference points for evaluating estimation accuracy. Furthermore, the results indicate that semi-blind channel estimation outperforms training-based methods, owing to its ability to exploit both training sequences and unknown data symbols.

\begin{figure*}[t]
\centering
\subfloat[Deterministic CRB under Assumption~\ref{assumption_dcrb_1}]{
\label{DCRB_Pilots}
\includegraphics[width=2.7 in]{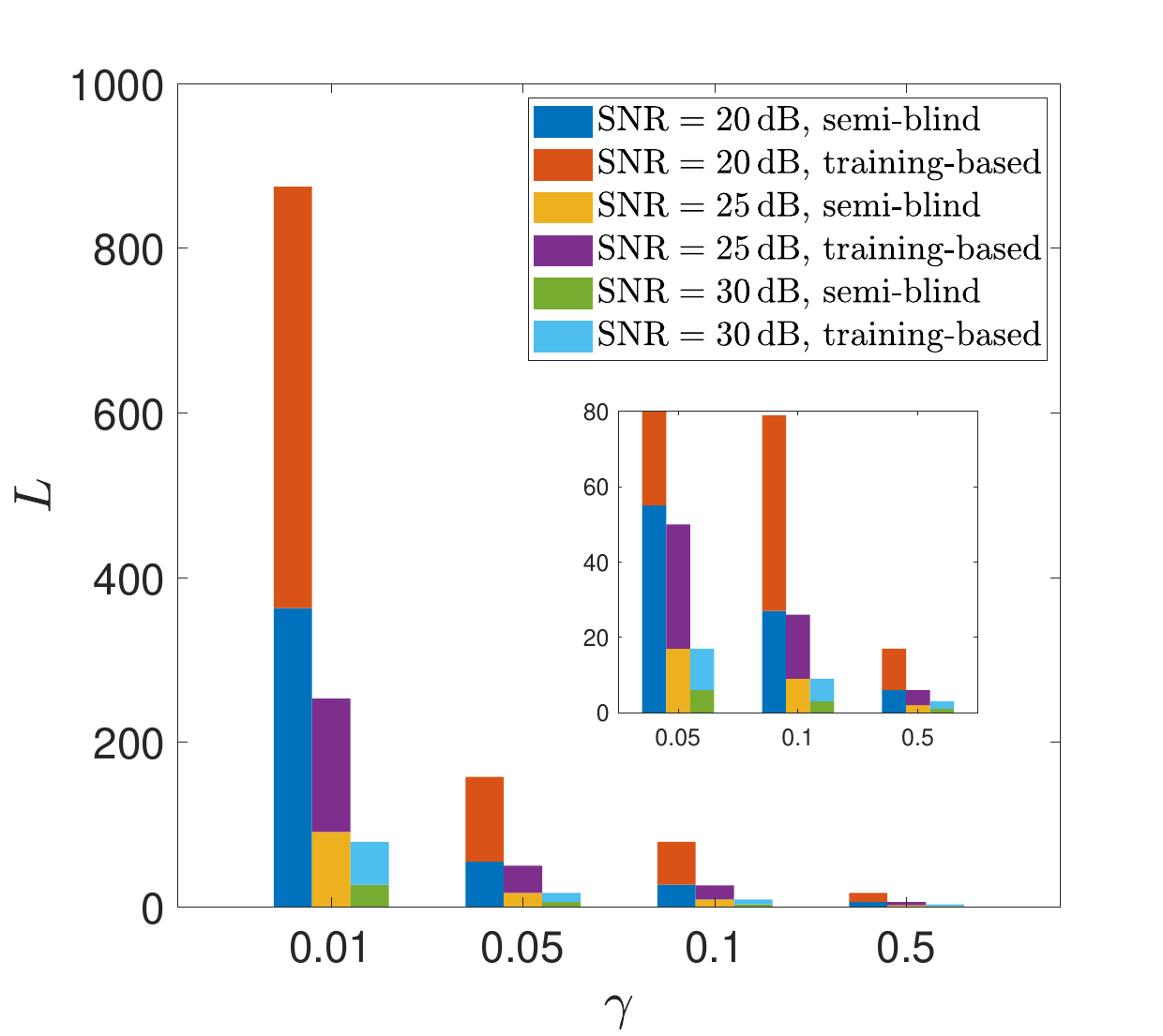}}
\quad
\subfloat[Stochastic CRB under Assumption~\ref{assumption_dcrb_1}]{
\label{SCRB_Pilots}
\includegraphics[width=2.7 in]{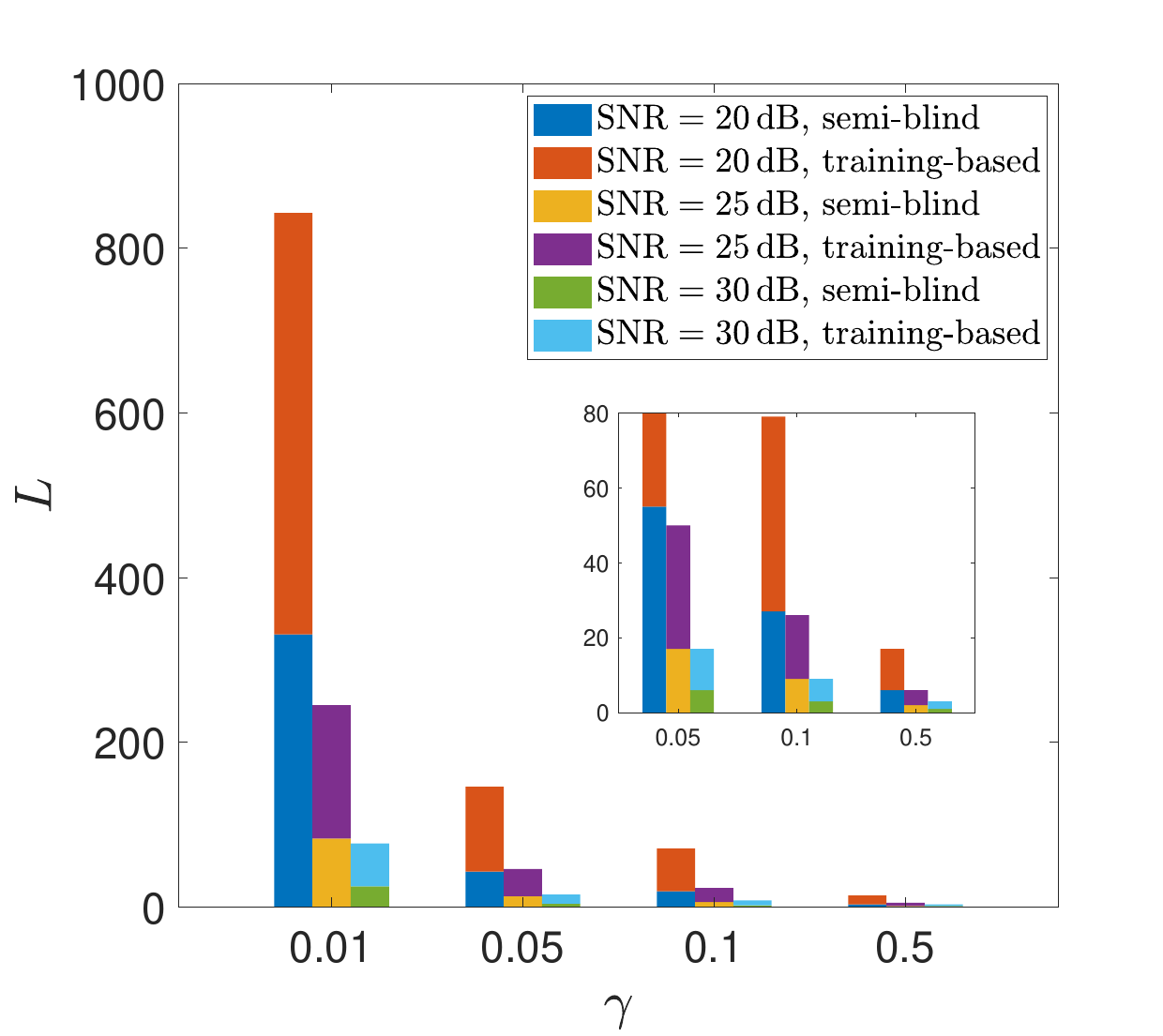}}
\caption{The length $N$ of the entire transmission block versus $\gamma$ with $\alpha=0.5$, and $K=M/2$.}
\label{Pilots_CRB}
\end{figure*}

\subsection{Experiment for Accuracy}
In this subsection, we verify the accuracy of the proposed asymptotic expression of deterministic CRB and stochastic CRB. Specifically, we compare the $\mathrm{NCAE}$ for the proposed asymptotic behavior and the corresponding expression in~\cite{nayebi2017semi} under Assumption~\ref{assumption_dcrb_1} for deterministic and stochastic CRB, respectively. In this setup, the number of users $K$ is half of the number of antennas  $M$, and the SNR is fixed at $10$ dB. Fig.~\ref{Accuracy_CRB} shows that for both the deterministic and stochastic CRBs, the error in our proposed expressions decreases as the system dimensions increase for all considered values of $\alpha$ and $\beta$, as specified in Assumption \ref{assumption_dcrb_1}. In contrast, the error in the expression from \cite{nayebi2017semi} exhibits a floor, highlighting its inaccuracy in approximating the CRBs.

\par
Next, we go on to validate the accuracy of the average asymptotic deterministic CRB when $K$ and $L$ are fixed as specified in Assumption \ref{assumption_DCRB_2}-Case 1. Fig.~\ref{SCRB_SNR} investigates the influence of $\mathrm{SNR}$ on the accuracy quality of the expression in \eqref{assumption_DCRB_2}. It can be seen that there exists an optimal value of $\mathrm{SNR}$ that minimizes the $\mathrm{NCAE}$. In addition, as the length of the whole transmission block increases, the $\mathrm{NCAE}$ becomes smaller.  

\subsection{Experiments for CRB}
In Fig.~\ref{DCRB_3}, we examine the behavior of the deterministic CRB as a function of the total transmission block length \( N \), under the regime defined by Assumption \ref{assumption_DCRB_2}, with \( \alpha = \frac{M}{N} = 0.5 \), \( K = 10 \), and either \( \frac{L}{N} = 0.5 \) or \( L \) fixed at 10. As observed, as \( N \) tends to infinity, the deterministic CRB converges to the asymptotic value from Theorem \ref{th:crb_simple} in both cases, \( L = 10 \) and \( \frac{L}{N} = 0.5 \). Moreover, in line with our expectations, when \( L \) is fixed, increasing \( N \) does not yield further improvement in the CRB. Additionally, while the results in \cite{nayebi2017semi} offer good accuracy when both \( N \) and \( \frac{L}{N} \) grow large, they fail to capture the correct behavior of the CRB when \( L \) is fixed at 10.

Fig.~\ref{CRB} describes the asymptotic behaviors of the stochastic CRB under the correlated channel model in Assumption \ref{ass:channel_model_cl} and the asymptotic regime in Assumption \ref{assumption_DCRB_2} with $\mathrm{SNR}=10$\,dB,  $K=8$, and $\alpha=0.5$ and $L=12$ for Assumption~\ref{assumption_DCRB_2}-Case 1 (Fig.\ref{CRB}-a), and $\beta=0.25$ for Assumption~\ref{assumption_DCRB_2}-Case 2 (Fig.\ref{CRB}-b). 
We observe that the CRB reaches its minimum when the channels across users have the same gains. Furthermore, consistent with our findings in Fig.~\ref{DCRB_3}, the result from \cite{nayebi2017semi} becomes inaccurate when $L$ is fixed. In this scenario, leveraging additional unknown data symbols for channel estimation does not lead to further improvements in estimation quality.

\subsection{Experiments for Required Pilots}
Finally, we demonstrate the practical application of the CRB derivations, which provide valuable guidance on how to reasonably select the number of pilot symbols. Given a target CRB value $\gamma$, the derived CRB expressions (8), (40), (17), and (45) can be used to determine the number of pilot symbols required to achieve this target CRB for both training-based and semi-blind techniques, with all other system parameters $M$, $N$, and $K$ held constant.

Fig.~\ref{Pilots_CRB} shows the required number of pilot symbols for various $\mathrm{SNR}$ levels and target CRB values $\gamma$, considering both stochastic and deterministic CRBs. The parameters are set as $N=2^{10}$, $M=\alpha N$, and $K=M/2$ with $\alpha =0.5$. As expected, with increasing $\mathrm{SNR}$ or higher target CRB values, the number of required pilots decreases for both approaches. Notably, semi-blind techniques require less than half the pilot symbols compared to training-based methods to achieve the same target CRB.

\par
\section{Conclusion}
\label{sec_con}
In this paper, we derived mathematically tractable expressions for the deterministic and stochastic CRBs for semi-blind channel estimation in massive MIMO systems. Using these expressions, we explored the asymptotic behavior of the CRBs under various growth regimes, considering the evolution of the number of known training symbols, data symbols, antennas at the base station, and the number of users. Our analysis revealed that the CRBs can approach zero as the total block length increases, but only if the number of training symbols grows at the same rate. Otherwise, the CRBs remain bounded below by a constant, which depends on the ratio between the number of training symbols and the number of users.

\appendices

\section{Derivation of the deterministic CRB}\label{derivation_DCRB}
The log-likelihood function of the receive signal $\mathbf{y}(n)$ can be expressed as
\begin{align}
    \mathcal{L} = \mathrm{const} - \sum_{n=0}^{N-1}\frac{1}{\sigma_v^2}\left\|\mathbf{y}(n)-\mathbf{Gs}(n)\right\|_2^2. 
\end{align}
We define a complex parameter vector as $\pmb{\xi} = \left[\mathbf{s}^T(L),\ldots,\mathbf{s}^T(N-1),\mathbf{g}_1^T,\ldots,\mathbf{g}_K^T\right]^T$, where $\mathbf{g}_k$ is the $k$th column of $\mathbf{G}$. For $n=L,\ldots,N-1$, we can obtain
\begin{align}
    &\frac{\partial\mathcal{L}}{\partial\mathbf{s}(n)} = \frac{1}{\sigma_v^2}\mathbf{v}(n)^H\mathbf{G}, \quad\frac{\partial\mathcal{L}}{\partial\mathbf{g}_k} = \frac{1}{\sigma_v^2}\sum_{n=0}^{N-1}s_k(n)\mathbf{v}(n)^H,\\
    &\frac{\partial\mathcal{L}}{\partial\mathbf{s}(n)^{\ast}} = \frac{1}{\sigma_v^2}\mathbf{v}(n)^T\mathbf{G}^{\ast},\,\,\frac{\partial\mathcal{L}}{\partial\mathbf{g}_k^{\ast}} = \frac{1}{\sigma_v^2}\sum_{n=0}^{N-1}s_k^{\ast}(n)\mathbf{v}(n)^T. 
\end{align}
Utilizing $\mathbb{E}\{\mathbf{v}(n)\mathbf{v}(p)^H\}=\sigma_v^2\mathbf{I}_M\delta(n-p)$, for $n,m=L,\ldots,N-1$ and $k,l=1,\ldots,K$, we can get
\begin{align}
    \pmb{\mathcal{J}}_{ss} =& \mathbb{E}\left\{\left(\frac{\partial\mathcal{L}}{\partial\mathbf{s}(n)}\right)^H\frac{\partial\mathcal{L}}{\partial\mathbf{s}(m)}\right\} =\mathbf{A}\delta(n-m),  \\
    \pmb{\mathcal{J}}_{gs} =& \mathbb{E}\left\{\left(\frac{\partial\mathcal{L}}{\partial\mathbf{g}_k}\right)^H\frac{\partial\mathcal{L}}{\partial\mathbf{s}(n)}\right\} = \mathbf{W}_{k}(n), \\
    \pmb{\mathcal{J}}_{gg} =& \mathbb{E}\left\{\left(\frac{\partial\mathcal{L}}{\partial\mathbf{g}_k}\right)^H\frac{\partial\mathcal{L}}{\partial\mathbf{g}_l}\right\} = \lambda_{kl}\mathbf{I}_M,\\
    \pmb{\mathcal{J}}_{ss^{\ast}} =& \pmb{\mathcal{J}}_{gs^{\ast}} = \pmb{\mathcal{J}}_{gg^{\ast}} = \pmb{0}, 
\end{align}
where $\mathbf{A}=\frac{1}{\sigma_v^2}\mathbf{G}^H\mathbf{G}$, $\mathbf{W}_{k}(n) = \frac{1}{\sigma_v^2}s_k(n)\mathbf{G}^H$, and $\lambda_{kl} = \frac{1}{\sigma_v^2}\sum_{n=0}^{N-1}s_k^{\ast}(n)s_l(n)$. And the expectation is taken over the distribution of the noise. 
Then we can get
\begin{align}
    \pmb{\mathcal{J}}_{\xi\xi} =& \mathbb{E}\left\{\left(\frac{\partial\mathcal{L}}{\partial\pmb{\xi}}\right)^H\frac{\partial\mathcal{L}}{\partial\pmb{\xi}}\right\} = \left[\begin{array}{cc}
        \pmb{\mathcal{J}}_{ss} & \pmb{\mathcal{J}}_{gs}^H \\
        \pmb{\mathcal{J}}_{gs} & \pmb{\mathcal{J}}_{gg} 
        \end{array}\right], \\
    \pmb{\mathcal{J}}_{\xi\xi^{\ast}} =& \mathbb{E}\left\{\left(\frac{\partial\mathcal{L}}{\partial\pmb{\xi}}\right)^H\frac{\partial\mathcal{L}}{\partial\pmb{\xi}^{\ast}}\right\} = \left[\begin{array}{cc}
        \pmb{\mathcal{J}}_{ss^{\ast}} & \pmb{\mathcal{J}}_{gs^{\ast}}^H \\
        \pmb{\mathcal{J}}_{gs^{\ast}} & \pmb{\mathcal{J}}_{gg^{\ast}} 
        \end{array}\right] =0,
\end{align}
Since $\pmb{\mathcal{J}}_{\xi\xi^{\ast}} =0$, the complex CRB defined as the inverse of the complex Fisher-information matrix $\pmb{\mathcal{J}}_{\xi\xi}$  provides the same lower bound on the covariance of unbiased estimators as the real CRB associated with the parameter \cite{de1997cramer}:\begin{align}
    \pmb{\xi}_R =& \left[\mathrm{Re}\{\mathbf{s}(L)\}^T,\mathrm{Im}\{\mathbf{s}(L)\}^T,\ldots,\mathrm{Re}\{\mathbf{s}(N-1)\}^T,\right.\notag\\
    &\,\,\,\mathrm{Im}\{\mathbf{s}(N-1)\}^T,\mathrm{Re}\{\mathbf{g}_1\}^T,\mathrm{Im}\{\mathbf{g}_1\}^T,\ldots,\notag\\
    &\left.\,\mathrm{Re}\{\mathbf{g}_K\}^T,\mathrm{Im}\{\mathbf{g}_K\}^T\right]^T,
\end{align}
Then the CRB of unknown data symbols and channel coefficients can be given by 
\begin{align}
    \mathbf{CRB}(\mathbf{S}_d,\mathbf{G}) =& \pmb{\mathcal{J}}_{\xi\xi}^{-1} = \left[\begin{array}{ccc:c}
         \mathbf{A} & & \pmb{0} & \mathbf{\Omega}_{L} \\
           & \ddots & & \vdots \\
         \pmb{0} & & \mathbf{A} & \mathbf{\Omega}_{N-1} \\
         \hdashline 
         \mathbf{\Omega}_L^H & \cdots & \mathbf{\Omega}_{N-1}^H & \mathbf{\Lambda}
            \end{array}\right]^{-1}. \notag
\end{align}
where $\mathbf{\Lambda}$ and $\mathbf{\Omega}_n = \left[\mathbf{W}_1(n),\cdots,\mathbf{W}_K(n)\right]$ for $n=L,\ldots,N-1$ are defined in (2) and (3), respectively. Based on the inverse of a block matrix, we can obtain
\begin{align}
       \mathbf{CRB}(\mathbf{G}) =& \left(\mathbf{\Lambda}-\sum_{n=L}^{N-1}\mathbf{\Omega}_n^H\mathbf{A}^{-1}\mathbf{\Omega}_n\right)^{-1}.
\end{align}

\par

\section{Stochastic CRB}
\subsection{Derivation of the stochastic CRB (Proof of Theorem \it{5})}
\label{app:stochastic_crb_derivation}
The log-likelihood function of $\mathbf{y}(n)$ is given by
\begin{align}
    \mathcal{L} =& \mathrm{const} - \sum_{n=0}^{L-1}\underbrace{\frac{1}{\sigma_v^2}\left\|\mathbf{y}(n)-\mathbf{Gs}(n)\right\|_2^2}_{\mathcal{L}_1(n)} \notag \\
     &- \sum_{n=L}^{N-1}\underbrace{\left(\mathbf{y}(n)^H\mathbf{R}^{-1}\mathbf{y}(n) + \mathrm{log}\,\mathrm{det}(\mathbf{R})\right)}_{\mathcal{L}_2(n)}. 
\end{align}
where $\mathcal{L}_1$ and $\mathcal{L}_2$ are related to pilot sequences and data symbols, respectively. We define a parameter vector as $\pmb{\alpha}=[\bar{\mathbf{g}}_1^T,\ldots,\bar{\mathbf{g}}_K^T,\Tilde{\mathbf{g}}_1^T,\ldots,\Tilde{\mathbf{g}}_K^T]$, the $(i,j)$th element of CRB, ($i,j=1,\cdots,2KM$), is given by 
\begin{align}       
    \left[\mathbf{CRB}^{-1}\right]_{ij} =& \sum_{n=0}^{L-1}\mathbb{E}\left\{\frac{\partial\mathcal{L}_1(n)}{\partial[\pmb{\alpha}]_i}\frac{\partial\mathcal{L}_1(n)}{\partial[\pmb{\alpha}]_j}\right\} \notag \\
    &+ \sum_{n=L}^{N-1}\mathbb{E}\left\{\frac{\partial\mathcal{L}_2(n)}{\partial[\pmb{\alpha}]_i}\frac{\partial\mathcal{L}_2(n)}{\partial[\pmb{\alpha}]_j}\right\}, 
\end{align}
where the expectation is taken with respect to the distribution of the noise and the data symbols $s_k(n), k=1,\cdots, K$ and $n=L,\cdots, N-1$, i.e.,
\begin{align}
     \mathbb{E}\left\{\frac{\partial\mathcal{L}_1(n)}{\partial[\bar{\mathbf{g}}_k]_p}\frac{\partial\mathcal{L}_1(n)}{\partial[\bar{\mathbf{g}}_l]_q}\right\} =& \mathbb{E}\left\{\frac{\partial\mathcal{L}_1(n)}{\partial[\tilde{\mathbf{g}}_k]_p}\frac{\partial\mathcal{L}_1(n)}{\partial[\tilde{\mathbf{g}}_l]_q}\right\} \notag \\
     =& \frac{2}{\sigma_v^2}\mathrm{Re}\left\{s_k(n)^{\ast}s_l(n)\right\}\delta(p-q), \notag\\
    \mathbb{E}\left\{\frac{\partial\mathcal{L}_1(n)}{\partial[\bar{\mathbf{g}}_k]_p}\frac{\partial\mathcal{L}_1(n)}{\partial[\tilde{\mathbf{g}}_l]_q}\right\} =& \frac{2}{\sigma_v^2}\mathrm{Im}\left\{s_k(n)^{\ast}s_l(n)\right\}\delta(p-q), \notag \\
     \mathbb{E}\left\{\frac{\partial\mathcal{L}_2(n)}{\partial[\bar{\mathbf{g}}_k]_p}\frac{\partial\mathcal{L}_2(n)}{\partial[\bar{\mathbf{g}}_l]_q}\right\} =& 2P_s^2\mathrm{Re}\left\{\mathrm{tr}\left(\mathbf{R}^{-1}\mathbf{T}_{kp}\mathbf{R}^{-1}\mathbf{T}_{lq}\right)\right\} \notag \\
    &+ 2P_s^2\mathrm{Re}\left\{\mathrm{tr}\left(\mathbf{R}^{-1}\mathbf{T}_{kp}\mathbf{R}^{-1}\mathbf{T}_{lq}^H\right)\right\} \notag \\
    =& 2P_s^2\mathrm{Re}\left\{\mathbf{g}_l^H\mathbf{R}^{-1}\mathbf{e}_p^T\mathbf{g}_k^H\mathbf{R}^{-1}\mathbf{e}_q^T\right\} \notag \\
    &+ 2P_s^2\mathrm{Re}\left\{\mathbf{e}_q\mathbf{R}^{-1}\mathbf{e}_p^T\mathbf{g}_k^H\mathbf{R}^{-1}\mathbf{g}_l\right\},\notag \\
    \mathbb{E}\left\{\frac{\partial\mathcal{L}_2(n)}{\partial[\bar{\mathbf{g}}_k]_p}\frac{\partial\mathcal{L}_2(n)}{\partial[\Tilde{\mathbf{g}}_l]_q}\right\}=& -2P_s^2\mathrm{Im}\left\{\mathrm{tr}\left(\mathbf{R}^{-1}\mathbf{T}_{kp}\mathbf{R}^{-1}\mathbf{T}_{lq}\right)\right\} \notag \\
    &+ 2P_s^2\mathrm{Im}\left\{\mathrm{tr}\left(\mathbf{R}^{-1}\mathbf{T}_{kp}\mathbf{R}^{-1}\mathbf{T}_{lq}^H\right)\right\} \notag\\ 
    =& -2P_s^2\mathrm{Im}\left\{\mathbf{g}_l^H\mathbf{R}^{-1}\mathbf{e}_p^T\mathbf{g}_k^H\mathbf{R}^{-1}\mathbf{e}_q^T\right\} \notag \\
    &+ 2P_s^2\mathrm{Im}\left\{\mathbf{e}_q\mathbf{R}^{-1}\mathbf{e}_p^T\mathbf{g}_k^H\mathbf{R}^{-1}\mathbf{g}_l\right\},\notag \\
     \mathbb{E}\left\{\frac{\partial\mathcal{L}_2(n)}{\partial[\Tilde{\mathbf{g}}_k]_p}\frac{\partial\mathcal{L}_2(n)}{\partial[\bar{\mathbf{g}}_l]_q}\right\}
     =& -2P_s^2\mathrm{Im}\left\{\mathrm{tr}\left(\mathbf{R}^{-1}\mathbf{T}_{kp}\mathbf{R}^{-1}\mathbf{T}_{lq}\right)\right\} \notag \\
     &-2P_s^2\mathrm{Im}\left\{\mathrm{tr}\left(\mathbf{R}^{-1}\mathbf{T}_{kp}\mathbf{R}^{-1}\mathbf{T}_{lq}^H\right)\right\} \notag \\
    =& -2P_s^2\mathrm{Im}\left\{\mathbf{g}_l^H\mathbf{R}^{-1}\mathbf{e}_p^T\mathbf{g}_k^H\mathbf{R}^{-1}\mathbf{e}_q^T\right\} \notag \\
    &- 2P_s^2\mathrm{Im}\left\{\mathbf{e}_q\mathbf{R}^{-1}\mathbf{e}_p^T\mathbf{g}_k^H\mathbf{R}^{-1}\mathbf{g}_l\right\}, \notag\\
    \mathbb{E}\left\{\frac{\partial\mathcal{L}_2(n)}{\partial[\Tilde{\mathbf{g}}_k]_p}\frac{\partial\mathcal{L}_2(n)}{\partial[\Tilde{\mathbf{g}}_l]_q}\right\}
    =& -2P_s^2\mathrm{Re}\left\{\mathrm{tr}\left(\mathbf{R}^{-1}\mathbf{T}_{kp}\mathbf{R}^{-1}\mathbf{T}_{lq}\right)\right\}  \notag \\
    &+ 2P_s^2\mathrm{Re}\left\{\mathrm{tr}\left(\mathbf{R}^{-1}\mathbf{T}_{kp}\mathbf{R}^{-1}\mathbf{T}_{lq}^H\right)\right\} \notag \\
    =& -2P_s^2\mathrm{Re}\left\{\mathbf{g}_l^H\mathbf{R}^{-1}\mathbf{e}_p^T\mathbf{g}_k^H\mathbf{R}^{-1}\mathbf{e}_q^T\right\} \notag \\
    &+ 2P_s^2\mathrm{Re}\left\{\mathbf{e}_q\mathbf{R}^{-1}\mathbf{e}_p^T\mathbf{g}_k^H\mathbf{R}^{-1}\mathbf{g}_l\right\},\notag
\end{align}
with $\mathbf{T}_{kp} =[\pmb{0}_{M\times(p-1)},\mathbf{g}_k,\pmb{0}_{M\times(M-p)}]^H$, $\mathbf{T}_{lq} =[\pmb{0}_{M\times(q-1)},\mathbf{g}_l,\pmb{0}_{M\times(M-q)}]^H$, $\mathbf{e}_p =[\mathbf{0}_{1\times(p-1)},1,\mathbf{0}_{1\times(M-p)}]$, and $\mathbf{e}_q =[\mathbf{0}_{1\times(q-1)},1,\mathbf{0}_{1\times(M-q)}]$. Here, $l=1,\ldots,K$, and $p,q=1,\ldots,M$.

Assume pilot sequences are orthogonal, i.e., $\sum_{n=0}^{L-1}s_k^{\ast}(n)s_l(n)=PL\delta(k-l)$, we derive the stochastic CRB as follows:
\begin{align*}
    \mathbf{CRB} = \left(\frac{2PL}{\sigma_v^2}\mathbf{I}_{2KM}+2P_s^2(N-L)(\mathcal{R}_1(\mathbf{T})+\mathcal{R}_2(\mathbf{C}))\right)^{-1},
\end{align*}
where $\mathcal{R}_1(\mathbf{T})$ and $\mathcal{R}_2(\mathbf{C})$ are defined in (33).

\par

\subsection{Simplification of the stochastic CRB (Proof of Proposition ~\ref{prop:stochastic_Crb})}
\label{app:simplified_crb}
We first provide the following preparatory results which are used in proving Theorem~\ref{prop:stochastic_Crb}.
\begin{lemma}
    $\mathbf{T}=(\mathbf{V}_G\otimes\mathbf{U}_G^{\ast})\mathbf{T}_G(\mathbf{V}_G^H\otimes\mathbf{U}_G^T)$.
    where ${\bf T}_G$ is given in \eqref{eq:TG}. 
    \begin{proof}
        Based on the  SVD of $\mathbf{G}$, $\mathbf{T}$ can be rewritten as: 
        \begin{align}
        \mathbf{T} =& \mathbf{V}_{G}\mathbf{D}_{G}^H\mathbf{R}_G^{-1}\mathbf{D}_{G}\mathbf{V}_{G}^H \otimes \mathbf{U}_{G}^{\ast}\mathbf{R}_G^{-1}\mathbf{U}_G^T  \notag\\
        =& (\mathbf{V}_G\otimes\mathbf{U}_G^{\ast})\mathbf{T}_G(\mathbf{V}_G^H\otimes\mathbf{U}_G^T).
        \end{align}
    \end{proof}
\end{lemma}

\begin{lemma}
   Matrix $\mathbf{C}$ can be written as ${\bf C}=(\mathbf{V}_G\otimes\mathbf{U}_G^{\ast})\mathbf{C}_G(\mathbf{V}_G^T\otimes\mathbf{U}_G^H)$.
   where ${\bf C}_G$ is given in \eqref{eq:CG}. 
    \begin{proof}
        Based on SVD of $\mathbf{G}$, we can derive
        \begin{align}
        (\mathbf{V}_G\otimes\mathbf{U}_G^{\ast})\mathbf{C}_G(\mathbf{V}_G^T\otimes\mathbf{U}_G^H) = \left[\begin{array}{ccc}
        \bar{\mathbf{C}}_{11} & \cdots & \bar{\mathbf{C}}_{1K} \\
        \vdots & \ddots & \vdots \\
        \bar{\mathbf{C}}_{K1} & \cdots & \bar{\mathbf{C}}_{KK} \\
        \end{array}\right],\notag
        \end{align}
        where 
        \begin{align}
            \bar{\mathbf{C}}_{ij} =& \mathbf{U}_G^{\ast}\sum_{p=1}^K\sum_{q=1}^Kv_{1q}\mathbf{C}_{G_{pq}}v_{1p}\mathbf{U}_G^H \notag \\
            =& \mathbf{U}_G^{\ast}\mathbf{R}_G^{-1}\sum_{p=1}^Kv_{jp}\mathbf{d}_p^{\ast}\sum_{q=1}^Kv_{iq}\mathbf{d}_q^H\mathbf{R}_G^{-1}\mathbf{U}_G^H 
        \end{align}
        with $v_{jp}=[\mathbf{V}]_{j,p}$ and $v_{iq}=[\mathbf{V}]_{i,q}$. Additionally, we note that:
        \begin{align}
            \mathbf{C}_{ij} =& \mathbf{U}_G^{\ast}\mathbf{R}_G^{-1}\mathbf{U}_G^T\mathbf{g}_j^{\ast}\mathbf{g}_i^H\mathbf{U}_G\mathbf{R}_G^{-1}\mathbf{U}_G^H \notag \\
            =& \mathbf{U}_G^{\ast}\mathbf{R}_G^{-1}\sum_{p=1}^Kv_{jp}\mathbf{d}_p^{\ast}\sum_{q=1}^Kv_{iq}\mathbf{d}_q^H\mathbf{R}_G^{-1}\mathbf{U}_G^H =\bar{\mathbf{C}}_{ij}.
        \end{align}
        Therefore, $\mathbf{C} = (\mathbf{V}_G\otimes\mathbf{U}_G^{\ast})\mathbf{C}_G(\mathbf{V}_G^T\otimes\mathbf{U}_G^H)$. 
    \end{proof}
\end{lemma}

\textit{Proof of Position~\ref{tr(CRB)=tr(A1)+tr(A2)}}: 
With the above lemmas at hand, we are now in a position to complete the proof of Position~\ref{tr(CRB)=tr(A1)+tr(A2)}. For that, we exploit the relations in ~\cite[page 71 and page 72]{bhatia2013matrix} to obtain:
    \begin{align}
        \mathcal{R}_1(\mathbf{T}) =&  \mathcal{R}_1((\mathbf{\mathbf{V}_G\otimes\mathbf{U}_G^{\ast}})\mathbf{T}_G(\mathbf{\mathbf{V}_G\otimes\mathbf{U}_G^{\ast}})^{-1}) \notag \\ 
        =& \mathcal{R}_1(\mathbf{\mathbf{V}_G\otimes\mathbf{U}_G^{\ast}})\mathcal{R}_1(\mathbf{T}_G)\mathcal{R}_1(\mathbf{\mathbf{V}_G\otimes\mathbf{U}_G^{\ast}})^{-1}, \\
        \mathcal{R}_2(\mathbf{C}) =&  \mathcal{R}_2((\mathbf{\mathbf{V}_G\otimes\mathbf{U}_G^{\ast}})\mathbf{C}_G(\mathbf{V}_G^{\ast}\otimes\mathbf{U}_G)^{-1}) \notag \\ 
        =& \mathcal{R}_2(\mathbf{\mathbf{V}_G\otimes\mathbf{U}_G^{\ast}})\mathcal{R}_2(\mathbf{I}_{MK})\mathcal{R}_2(\mathbf{C}_G) \notag \\
        &\cdot(\mathcal{R}_2(\mathbf{\mathbf{V}_G\otimes\mathbf{U}_G^{\ast}})\mathcal{R}_2(\mathbf{I}_{MK}))^{-1} \notag \\
        =& \mathcal{R}_1(\mathbf{\mathbf{V}_G\otimes\mathbf{U}_G^{\ast}})\mathcal{R}_2(\mathbf{C}_G)\mathcal{R}_1(\mathbf{\mathbf{V}_G\otimes\mathbf{U}_G^{\ast}})^{-1},
    \end{align}
   We can thus write:
    \begin{align}
        \mathcal{R}_1(\mathbf{T}) + \mathcal{R}_2(\mathbf{C}) =& \mathcal{R}_1(\mathbf{\mathbf{V}_G\otimes\mathbf{U}_G^{\ast}})(\mathcal{R}_1(\mathbf{T}_G)+\mathcal{R}_2(\mathbf{C}_G)) \notag \\
        &\cdot\mathcal{R}_1(\mathbf{\mathbf{V}_G\otimes\mathbf{U}_G^{\ast}})^{-1}.
    \end{align}
    As a result, the eigenvalues   of $\mathcal{R}_1(\mathbf{T}) + \mathcal{R}_2(\mathbf{C})$ coincide with those of  $\begin{bmatrix} {\bf T}_G+{\bf C}_{G} & {\bf 0}\\ {\bf 0 }& 
 {\bf T}_G-{\bf C}_G\end{bmatrix}$. Thus, 
    \begin{align}
    \mathrm{tr}(\mathbf{CRB})
    =& \sum_{l=1}^{KM}\frac{1}{\frac{2L}{\sigma_v^2}+2P_s^2(N-L)\lambda_l(\mathbf{T}_G+\mathbf{C}_G)} \notag \\
    &+ \sum_{l=1}^{KM}\frac{1}{\frac{2L}{\sigma_v^2}+2P_s^2(N-L)\lambda_l(\mathbf{T}_G-\mathbf{C}_G)} \notag \\
    =& \mathrm{tr}(\mathbf{A}_1)+\mathrm{tr}(\mathbf{A}_2). 
    \end{align}

\par

\subsection{Derivation of the average stochastic CRB (Proof of Theorem \it{6})} 
\label{app:stochastic_crb}
Matrix ${\bf C}_G$ has at most one non-zero element in each row and each column. Consequently, the non-zero elements of matrix ${\bf A}_1$ consist of the $KM$ diagonal elements, along with $K(K-1)$ off-diagonal elements produced by the block matrices ${\bf C}_{G_{ij}}$ for $i\neq j$.

For $i=1,\cdots K$ and $j=1,\cdots, M$ let $a_{ij}$ denote the $(i-1)M+j$ th diagonal element of matrix ${\bf A}_1$. Then, ${a}_{ij}$ is given by:
\begin{align}
    a_{ij}= \left\{\begin{array}{ll}
        \frac{2PL}{\sigma_v^2} + \frac{4P_s^2(N-L)\sigma_i^2}{(P_s\sigma_i^2+\sigma_v^2)^2}, & \text{if } j = i, \\
        \frac{2PL}{\sigma_v^2} + \frac{2P_s^2(N-L)\sigma_i^2}{(P_s\sigma_i^2+\sigma_v^2)(P_s\sigma_j^2+\sigma_v^2)}, & \text{if } j\neq i \ \text{and} \  j \leq K, \\
        \frac{2PL}{\sigma_v^2} + \frac{2P_s^2(N-L)\sigma_i^2}{(P_s\sigma_i^2+\sigma_v^2)\sigma_v^2}, & \text{if } j>K.
    \end{array}\right.\notag
\end{align}
For $i = 1, \cdots, K$ and $j = 1, \cdots, M$ with $j \neq i$, denote by $\tilde{a}_{i,k}$ the element of matrix ${\bf A}_1$ at row $(i-1)M + j$ and column $k$. Then
\[
\tilde{a}_{i,k} = \left\{
\begin{array}{ll}
2P_s^2(N-L)c_{ij} & \text{if } k = (j-1)M + i \ \text{and} \ j \leq K, \\
0 & \text{otherwise.}
\end{array}
\right.
\]
It follows from the particular structure of matrix ${\bf A}_1$ that at each row and column there are at most two non-zero elements: one diagonal element and possibly one off-diagonal element resulting from the contribution of matrix ${\bf C}_G$. For $i = 1, \cdots, K$ and $j = 1, \cdots, M$, let $a_{ij}^{-1}$ denote the $(i-1)M + j$-th diagonal element of ${\bf A}_1^{-1}$, the inverse of matrix ${\bf A}_1$. Due to the particular structure of ${\bf A}_1$, the diagonal elements of its inverse are computed either by directly inverting the corresponding diagonal elements, or by considering the first diagonal element of the inverse of a $2 \times 2$ matrix. Specifically, for $i = 1, \cdots, K$ and $j = 1, \cdots, M$, $a_{ij}^{-1}$ is given by:
\[
a_{ij}^{-1} = \left\{
\begin{array}{ll}
\frac{1}{a_{ij}},  & \text{if } i = j \ \text{or} \ j > K, \\
\frac{a_{ji}}{a_{ji} a_{ij} -4P_s^4(N-L)^2 c_{ij}^2}, & \text{otherwise},
\end{array}
\right.
\]
where, for $j \leq K$ and $j \neq i$, the element $a_{ij}^{-1}$ is computed as the first diagonal element of the inverse of the $2 \times 2$ matrix $\begin{bmatrix} a_{ij} &2P_s^2(N-L) c_{ij} \\ 2P_s^2(N-L)c_{ij} & a_{ji} \end{bmatrix}$.
Similarly, we can derive the diagonal elements of matrix ${\bf A}_2$. For $i=1,\cdots,K$ and $j=1,\cdots,M$, let $b_{ij}$ the $(i-1)M+j$ th diagonal element of matrix ${\bf A}_2$. Then, $b_{ij}$ is given by:
\begin{align}
b_{ij}=\left\{
\begin{array}{ll}
\frac{2PL}{\sigma_v^2}, &\text{ if } j=i,\\
\frac{2PL}{\sigma_v^2}+\frac{2P_s^2(N-L)\sigma_i^2}{(P_s\sigma_i^2+\sigma_v^2)(P_s\sigma_j^2+\sigma_v^2)}, & \text{if } j\neq i  \ \text{and } j\leq K,\\
\frac{2PL}{\sigma_v^2}+\frac{2P_s^2(N-L)\sigma_i^2}{(P_s\sigma_i^2+\sigma_v^2)\sigma_v^2}, & \text{if } j>K,
\end{array}
\right. \notag
\end{align}
Moreover, letting $b_{ij}^{-1}$, the $(i-1)M+j$-th diagonal element of ${\bf A}_2^{-1}$ with $i=1,\cdots,K$ and $j=1,\cdots,M$, then $b_{ij}^{-1}$ writes as:
\[
b_{ij}^{-1} = \left\{
\begin{array}{ll}
\frac{1}{b_{ij}},  & \text{if } i = j \ \text{or} \ j > K, \\
\frac{b_{ji}}{b_{ji} b_{ij} -4P_s^4(N-L)^2 c_{ij}^2}, & \text{otherwise},
\end{array}
\right.
\]
Using these expressions, we expand $\frac{1}{K}\mathrm{tr}({\bf A}_1^{-1})$ and  $\frac{1}{K}\mathrm{tr}({\bf A}_2^{-1})$ to obtain the final expressions in \eqref{eq:A_1} and \eqref{eq:A_2}. 
\begin{figure*}[htbp]
\begin{align}
&\frac{\mathrm{tr}(\mathbf{A}_1^{-1})}{K} = \frac{1}{K}\sum_{i=1}^K \frac{1}{a_{ii}}+\frac{1}{K}\sum_{i=1}^K \sum_{j=K+1}^{M}\frac{1}{a_{ij}}+\frac{1}{K}\sum_{i=1}^K\sum_{\substack{j=1 \\ j\neq i}}^K\frac{a_{ji}}{a_{ij}a_{ji}-4P_s^4(N-L)^2c_{ij}^2} \notag\\
&=\frac{1}{K}\sum_{i=1}^K\sum_{j=1}^M \frac{1}{a_{ij}}+\frac{1}{K}\sum_{i=1}^K \sum_{\substack{j=1 \\ j\neq i}}^K\frac{4P_s^4(N-L)^2c_{ij}^2\frac{1}{a_{ij}}}{a_{ij}a_{ji}-4P_s^4(N-L)^2c_{ij}^2} =\frac{1}{K}\sum_{i=1}^K\left(\frac{1}{\frac{2PL}{\sigma_v^2}+\frac{4P_s^2(N-L)\sigma_i^2}{(P_s\sigma_i^2+\sigma_v^2)^2}} - \frac{1}{\frac{2PL}{\sigma_v^2}+\frac{2P_s^2(N-L)\sigma_i^2}{(P_s\sigma_i^2+\sigma_v^2)^2}}\right)\notag\\
&+ \frac{1}{K}\sum_{i=1}^K\left(\sum_{j=1}^K\frac{1}{\frac{2PL}{\sigma_v^2}+\frac{2P_s^2(N-L)\sigma_i^2}{(P_s\sigma_i^2+\sigma_v^2)(P_s\sigma_j^2+\sigma_v^2)}} + \frac{M-K}{\frac{2PL}{\sigma_v^2}+\frac{2P_s^2(N-L)\sigma_i^2}{(P_s\sigma_i^2+\sigma_v^2)\sigma_v^2}}\right)\notag \\
&+\frac{1}{K}\sum_{i=1}^K\sum_{j=1}^K \frac{1}{\Big(\frac{2PL}{\sigma_v^2}+\frac{2P_s^2(N-L)\sigma_i^2}{(P_s\sigma_i^2+\sigma_v^2)(P_s\sigma_j^2+\sigma_v^2)}\Big)}\frac{4P_s^4(N-L)^2\sigma_i^2\sigma_j^2}{\Big(\frac{4P^2L^2}{\sigma_v^4}+\frac{4PLP_s^2(N-L)}{\sigma_v^2}\frac{\sigma_i^2+\sigma_j^2}{(P_s\sigma_i^2+\sigma_v^2)(P_s\sigma_j^2+\sigma_v^2)}\Big)}\frac{1}{(P_s\sigma_i^2+\sigma_v^2)^2(P_s\sigma_j^2+\sigma_v^2)^2} \notag\\
&-\frac{1}{K}\sum_{i=1}^K\frac{2P_s^4(N-L)^2\sigma_i^4\sigma_v^2}{(P_s\sigma_i^2+\sigma_v^2)^4PL}\frac{1}{(\frac{2PL}{\sigma_v^2}+\frac{2P_s^2(N-L)\sigma_i^2}{(P_s\sigma_i^2+\sigma_v^2)^2})(\frac{2PL}{\sigma_v^2}+\frac{4P_s^2(N-L)\sigma_i^2}{(P_s\sigma_i^2+\sigma_v^2)^2})}, \label{eq:A_1}
\end{align}
\begin{align}
\frac{\mathrm{tr}({\bf A}_2^{-1})}{K}&=\frac{1}{K}\sum_{i=1}^K \frac{1}{b_{ii}}+\frac{1}{K}\sum_{i=1}^K \sum_{j=K+1}^{M}\frac{1}{b_{ij}}+\frac{1}{K}\sum_{i=1}^K\sum_{\substack{j=1 \\ j\neq i}}^K\frac{b_{ji}}{b_{ij}b_{ji}-4P_s^4(N-L)^2c_{ij}^2} \notag\\
&=\frac{1}{K}\sum_{i=1}^K \sum_{j=1}^{M}\frac{1}{b_{ij}}+\frac{1}{K}\sum_{i=1}^K\sum_{\substack{j=1\\ j\neq i}}^{K}\frac{\frac{1}{b_{ij}}c_{ij}^24P_s^4(N-L)^2}{b_{ij}b_{ji}-4P_s^4(N-L)^2c_{ij}^2} =\frac{\sigma_v^2}{2PL}+\frac{1}{K}\sum_{i=1}^K \sum_{j=1}^K \frac{1}{\frac{2PL}{\sigma_v^2}+\frac{2P_s^2(N-L)\sigma_i^2}{(P_s\sigma_i^2+\sigma_v^2)(P_s\sigma_j^2+\sigma_v^2)}} \notag\\
&-\frac{1}{K}\sum_{k=1}^K \frac{1}{\frac{2PL}{\sigma_v^2}+\frac{2P_s^2(N-L)\sigma_i^2}{(P_s\sigma_i^2+\sigma_v^2)^2}}+\frac{M-K}{K}\sum_{i=1}^K \frac{1}{\frac{2PL}{\sigma_v^2}+\frac{2P_s^2(N-L)\sigma_i^2}{(P_s\sigma_i^2+\sigma_v^2)\sigma_v^2}} +\frac{1}{K}\sum_{i=1}^K\sum_{j=1}^K \frac{1}{\Big(\frac{2PL}{\sigma_v^2}+\frac{2P_s^2(N-L)\sigma_i^2}{(P_s\sigma_i^2+\sigma_v^2)(P_s\sigma_j^2+\sigma_v^2)}\Big)}\notag\\
&\cdot\frac{4P_s^4(N-L)^2\sigma_i^2\sigma_j^2}{\Big(\frac{4P^2L^2}{\sigma_v^4}+\frac{4PLP_s^2(N-L)}{\sigma_v^2}\frac{\sigma_i^2+\sigma_j^2}{(P_s\sigma_i^2+\sigma_v^2)(P_s\sigma_j^2+\sigma_v^2)}\Big)}\frac{1}{(P_s\sigma_i^2+\sigma_v^2)^2(P_s\sigma_j^2+\sigma_v^2)^2} \notag\\
&-\frac{1}{K}\sum_{i=1}^K\frac{2P_s^4(N-L)^2\sigma_i^4\sigma_v^2}{(P_s\sigma_i^2+\sigma_v^2)^4PL}\frac{1}{(\frac{2PL}{\sigma_v^2}+\frac{2P_s^2(N-L)\sigma_i^2}{(P_s\sigma_i^2+\sigma_v^2)^2})(\frac{2PL}{\sigma_v^2}+\frac{4P_s^2(N-L)\sigma_i^2}{(P_s\sigma_i^2+\sigma_v^2)^2})}. \label{eq:A_2}
\end{align}
\hrule
\end{figure*}
Using
\begin{align}
&\frac{1}{K}\sum_{i=1}^K \left(\frac{1}{\frac{2PL}{\sigma_v^2}+\frac{4P_s^2(N-L)\sigma_i^2}{(P_s\sigma_i^2+\sigma_v^2)^2}}-\frac{1}{\frac{PL}{\sigma_v^2}+\frac{P_s^2(N-L)\sigma_i^2}{(P_s\sigma_i^2+\sigma_v^2)^2}}-\right.\notag\\
&\left.\frac{P_s^4(N-L)^2\sigma_i^4\sigma_v^2}{PL(P_s\sigma_i^2+\sigma_v^2)^4(\frac{PL}{\sigma_v^2}+\frac{P_s^2(N-L)\sigma_i^2}{(P_s\sigma_i^2+\sigma_v^2)^2})(\frac{PL}{\sigma_v^2}+\frac{2P_s^2(N-L)\sigma_i^2}{(P_s\sigma_i^2+\sigma_v^2)^2})}\right)\notag\\
&+\frac{\sigma_v^2}{2PL}=0, \label{eq:simp}
\end{align}
we simplify $\overline{\rm CRB}_s^{\rm avg}$ as in \eqref{eq:crb_avg}.
\begin{figure*}[htbp]
\begin{align}
\overline{\rm CRB}_s^{\rm avg}&=\frac{1}{K}\sum_{i=1}^K\sum_{j=1}^K \frac{1}{\frac{PL}{\sigma_v^2}+\frac{P_s^2(N-L)\sigma_j^2}{(P_s\sigma_i^2+\sigma_v^2)(P_s\sigma_j^2+\sigma_v^2)}}+\frac{M-K}{K}\sum_{i=1}^K\frac{1}{\frac{PL}{\sigma_v^2}+\frac{P_s^2(N-L)\sigma_i^2}{(P_s\sigma_i^2+\sigma_v^2)\sigma_v^2}}\notag\\
&+\frac{1}{K}\sum_{i=1}^K\sum_{j=1}^K \frac{P_s^4(N-L)^2\sigma_i^2\sigma_j^2\sigma_v^2}{PL(\frac{PL}{\sigma_v^2}+\frac{P_s^2(N-L)\sigma_i^2}{(P_s\sigma_i^2+\sigma_v^2)(P_s\sigma_j^2+\sigma_v^2)})(\frac{PL}{\sigma_v^2}+\frac{P_s^2(N-L)(\sigma_i^2+\sigma_j^2)}{(P_s\sigma_i^2+\sigma_v^2)(P_s\sigma_j^2+\sigma_v^2)}) }\frac{1}{(P_s\sigma_i^2+\sigma_v^2)^2(P_s\sigma_j^2+\sigma_v^2)^2}\label{eq:crb_avg}\\
&=\frac{M-K}{K}\sum_{i=1}^K \frac{1}{\frac{PL}{\sigma_v^2}+\frac{P_s^2(N-L)\sigma_i^2}{(P_s\sigma_i^2+\sigma_v^2)\sigma_v^2}}+\frac{1}{K}\sum_{i=1}^K\sum_{j=1}^K \frac{\frac{PL}{\sigma_v^2}+\frac{P_s^2(N-L)\sigma_j^2}{(P_s\sigma_j^2+\sigma_v^2)(P_s\sigma_i^2+\sigma_v^2)}}{PL(\frac{PL}{\sigma_v^4}+\frac{P_s^2(N-L)(\sigma_i^2+\sigma_j^2)}{\sigma_v^2(P_s\sigma_i^2+\sigma_v^2)(P_s\sigma_j^2+\sigma_v^2)})}.
\end{align}
\hrule
\end{figure*}
To facilitate asymptotic analysis of the stochastic average CRB, we shall express it in terms of the empirical Stieltjes transform of matrix ${\bf G}^{H}{\bf G}$ defined as:
$$
m_{{\bf G}^{H}{\bf G}}(x)=\frac{1}{K}\sum_{j=1}^K \frac{1}{\sigma_i^2-x}.
$$
Introducing this notation, we obtain the expression of the average stochastic CRB in (\ref{eq:stochastic_crb}). 


\bibliographystyle{IEEEtran}
\bibliography{references}

\begin{thebibliography}{10}
\providecommand{\url}[1]{#1}
\csname url@samestyle\endcsname
\providecommand{\newblock}{\relax}
\providecommand{\bibinfo}[2]{#2}
\providecommand{\BIBentrySTDinterwordspacing}{\spaceskip=0pt\relax}
\providecommand{\BIBentryALTinterwordstretchfactor}{4}
\providecommand{\BIBentryALTinterwordspacing}{\spaceskip=\fontdimen2\font plus
\BIBentryALTinterwordstretchfactor\fontdimen3\font minus \fontdimen4\font\relax}
\providecommand{\BIBforeignlanguage}[2]{{%
\expandafter\ifx\csname l@#1\endcsname\relax
\typeout{** WARNING: IEEEtran.bst: No hyphenation pattern has been}%
\typeout{** loaded for the language `#1'. Using the pattern for}%
\typeout{** the default language instead.}%
\else
\language=\csname l@#1\endcsname
\fi
#2}}
\providecommand{\BIBdecl}{\relax}
\BIBdecl

\bibitem{noh2014new}
S.~Noh, Y.~Sung, and M.~D. Zoltowski, ``A new precoder design for blind channel estimation in {MIMO-OFDM} systems,'' \emph{IEEE Transactions on Wireless Communications}, vol.~13, no.~12, pp. 7011--7024, 2014.

\bibitem{tong1994blind}
L.~Tong, G.~Xu, and T.~Kailath, ``{Blind identification and equalization based on second-order statistics: A time domain approach},'' \emph{IEEE Transactions on Information Theory}, vol.~40, no.~2, pp. 340--349, 1994.

\bibitem{hassibi2003much}
B.~Hassibi and B.~M. Hochwald, ``How much training is needed in multiple-antenna wireless links?'' \emph{IEEE Transactions on Information Theory}, vol.~49, no.~4, pp. 951--963, 2003.

\bibitem{srinivas2019iterative}
B.~Srinivas, K.~Mawatwal, D.~Sen, and S.~Chakrabarti, ``An iterative semi-blind channel estimation scheme and uplink spectral efficiency of pilot contaminated one-bit massive {MIMO} systems,'' \emph{IEEE Transactions on Vehicular Technology}, vol.~68, no.~8, pp. 7854--7868, 2019.

\bibitem{aldana2003channel}
C.~H. Aldana, E.~de~Carvalho, and J.~M. Cioffi, ``Channel estimation for multicarrier multiple input single output systems using the {EM} algorithm,'' \emph{IEEE Transactions on Signal Processing}, vol.~51, no.~12, pp. 3280--3292, 2003.

\bibitem{marzetta2010noncooperative}
T.~L. Marzetta, ``Noncooperative cellular wireless with unlimited numbers of base station antennas,'' \emph{IEEE transactions on wireless communications}, vol.~9, no.~11, pp. 3590--3600, 2010.

\bibitem{rusek2012scaling}
F.~Rusek, D.~Persson, B.~K. Lau, E.~G. Larsson, T.~L. Marzetta, O.~Edfors, and F.~Tufvesson, ``Scaling up {MIMO}: {Opportunities} and challenges with very large arrays,'' \emph{IEEE signal processing magazine}, vol.~30, no.~1, pp. 40--60, 2012.

\bibitem{nayebi2017semi}
E.~Nayebi and B.~D. Rao, ``Semi-blind channel estimation for multiuser massive {MIMO} systems,'' \emph{IEEE Transactions on Signal Processing}, vol.~66, no.~2, pp. 540--553, 2017.

\bibitem{al2021semi}
M.~Al-Shoukairi and B.~D. Rao, ``Semi-blind channel estimation in {MIMO} systems with discrete priors on data symbols,'' \emph{IEEE Signal Processing Letters}, vol.~29, pp. 51--54, 2021.

\bibitem{abuthinien2008semi}
M.~Abuthinien, S.~Chen, and L.~Hanzo, ``Semi-blind joint maximum likelihood channel estimation and data detection for {MIMO} systems,'' \emph{IEEE Signal Processing Letters}, vol.~15, pp. 202--205, 2008.

\bibitem{jagannatham2006whitening}
A.~K. Jagannatham and B.~D. Rao, ``Whitening-rotation-based semi-blind {MIMO} channel estimation,'' \emph{IEEE Transactions on Signal Processing}, vol.~54, no.~3, pp. 861--869, 2006.

\bibitem{rekik2024fast}
O.~Rekik, K.~N. Aliyu, B.~M. Tuan, K.~Abed-Meraim, and N.~L. Trung, ``Fast subspace-based blind and semi-blind channel estimation for {MIMO-OFDM} systems,'' \emph{IEEE Transactions on Wireless Communications}, 2024, early access.

\bibitem{lawal2023semi}
A.~Lawal, K.~Abed-Meraim, A.~Zerguine, Q.~Mayyala, K.~N. Aliyu, and A.~Muqaibel, ``Semi-blind signal estimation using {Toeplitz} structured-based subspace method,'' \emph{IEEE Wireless Communications Letters}, vol.~12, no.~8, pp. 1319--1323, 2023.

\bibitem{stoica1990performance}
P.~Stoica and A.~Nehorai, ``Performance study of conditional and unconditional direction-of-arrival estimation,'' \emph{IEEE Transactions on Acoustics, Speech, and Signal Processing}, vol.~38, no.~10, pp. 1783--1795, 1990.

\bibitem{sandkuhler1987accuracy}
U.~Sandkuhler and J.~Bohme, ``Accuracy of maximum-likelihood estimates for array processing,'' in \emph{ICASSP'87. IEEE International Conference on Acoustics, Speech, and Signal Processing}, vol.~12.\hskip 1em plus 0.5em minus 0.4em\relax IEEE, 1987, pp. 2015--2018.

\bibitem{de1997cramer}
E.~De~Carvalho and D.~T. Slock, ``Cramer-rao bounds for semi-blind, blind and training sequence based channel estimation,'' in \emph{First IEEE signal processing workshop on signal processing advances in wireless communications}, 1997, pp. 129--132.

\bibitem{de1997asymptotic}
E.~de~Carvalho and D.~Slock, ``Asymptotic performance of {ML} methods for semi-blind channel estimation,'' in \emph{Conference Record of the Thirty-First Asilomar Conference on Signals, Systems and Computers (Cat. No. 97CB36136)}, vol.~2, 1997, pp. 1624--1628.

\bibitem{bjornson2017massive}
E.~Bj{\"o}rnson, J.~Hoydis, L.~Sanguinetti \emph{et~al.}, ``Massive {MIMO} networks: Spectral, energy, and hardware efficiency,'' \emph{Foundations and Trends{\textregistered} in Signal Processing}, vol.~11, no. 3-4, pp. 154--655, 2017.

\bibitem{gao2007blind}
F.~Gao and A.~Nallanathan, ``Blind channel estimation for {MIMO OFDM} systems via nonredundant linear precoding,'' \emph{IEEE Transactions on Signal Processing}, vol.~55, no.~2, pp. 784--789, 2007.

\bibitem{shin2007blind}
C.~Shin, R.~W. Heath, and E.~J. Powers, ``Blind channel estimation for {MIMO-OFDM} systems,'' \emph{IEEE Transactions on Vehicular Technology}, vol.~56, no.~2, pp. 670--685, 2007.

\bibitem{ghavami2017blind}
K.~Ghavami and M.~Naraghi-Pour, ``Blind channel estimation and symbol detection for multi-cell massive {MIMO} systems by expectation propagation,'' \emph{IEEE Transactions on Wireless Communications}, vol.~17, no.~2, pp. 943--954, 2017.

\bibitem{silverstein1995empirical}
J.~W. Silverstein and Z.~Bai, ``On the empirical distribution of eigenvalues of a class of large dimensional random matrices,'' \emph{Journal of Multivariate analysis}, vol.~54, no.~2, pp. 175--192, 1995.

\bibitem{feyzmahdavian2012contractive}
H.~R. Feyzmahdavian, M.~Johansson, and T.~Charalambous, ``Contractive interference functions and rates of convergence of distributed power control laws,'' \emph{IEEE Transactions on Wireless Communications}, vol.~11, no.~12, pp. 4494--4502, 2012.

\bibitem{yates1995framework}
R.~D. Yates, ``A framework for uplink power control in cellular radio systems,'' \emph{IEEE Journal on selected areas in communications}, vol.~13, no.~7, pp. 1341--1347, 1995.

\bibitem{marchenko1967distribution}
V.~A. Marchenko and L.~A. Pastur, ``Distribution of eigenvalues for some sets of random matrices,'' \emph{Matematicheskii Sbornik}, vol. 114, no.~4, pp. 507--536, 1967.

\bibitem{couillet2011random}
R.~Couillet and M.~Debbah, \emph{Random Matrix Methods for Wireless Communications}.\hskip 1em plus 0.5em minus 0.4em\relax New York, NY, USA: Cambridge Univ. Presss, 2011.

\bibitem{forenza2007simplified}
A.~Forenza, D.~J. Love, and R.~W. Heath, ``Simplified spatial correlation models for clustered {MIMO} channels with different array configurations,'' \emph{IEEE Transactions on Vehicular Technology}, vol.~56, no.~4, pp. 1924--1934, 2007.

\bibitem{tang2001mobile}
A.~Tang, J.~Sun, and K.~Gong, ``Mobile propagation loss with a low base station antenna for {NLOS} street microcells in urban area,'' in \emph{IEEE VTS 53rd Vehicular Technology Conference, Spring 2001. Proceedings (Cat. No. 01CH37202)}, vol.~1.\hskip 1em plus 0.5em minus 0.4em\relax IEEE, 2001, pp. 333--336.

\bibitem{zhong2024subspace}
B.~Zhong, X.~Zhu, and E.~G. Lim, ``Subspace-based semi-blind channel estimation for user-centric cell-free massive {MIMO} systems,'' in \emph{2024 IEEE 99th Vehicular Technology Conference (VTC2024-Spring)}.\hskip 1em plus 0.5em minus 0.4em\relax IEEE, 2024, pp. 1--5.

\bibitem{bhatia2013matrix}
J.~Horn and C.~R. Johnson, \emph{Matrix analysis}.\hskip 1em plus 0.5em minus 0.4em\relax Cambridge, U.K.: Cambridge Univ. Press, 2013.

\end{thebibliography}
\end{document}